\documentclass[11pt]{article}

\usepackage[utf8]{inputenc}
\usepackage[english]{babel}
 
\usepackage[numbers,sort&compress]{natbib} 
 \usepackage{amssymb}
\usepackage{amsmath}

\usepackage{mathrsfs}

\usepackage{amsthm}

 \usepackage{hyperref}
\hypersetup{
   colorlinks   =  true,
    linkcolor    = blue,
    citecolor    = red,
     urlcolor	=blue,
 }

  \usepackage{doi}

\newtheorem{theorem}{Theorem}[section]

\newtheorem{lemma}[theorem]{Lemma}
\newtheorem{proposition}{Proposition}

\theoremstyle{definition} 

\newtheorem{definition}[theorem]{Definition}

\newcommand{\dd}{\text{d}}

\def\ham{{\mathcal{H}}}

\numberwithin{equation}{section}
\numberwithin{proposition}{section}
\numberwithin{example}{section}
\newcommand\be{\begin{equation}}
\newcommand\ee{\end{equation}}
\newcommand\bea{\begin{eqnarray}}
\newcommand\eea{\end{eqnarray}}
 
\usepackage{color}

\newcommand{\I}{I} % first integral
\newcommand{\mm}{M} % manifold
\newcommand{\mmm}{ {M}} % cal manifold
\newcommand{\C}{\mathcal{C}} % Casimir

\newcommand{\mS}{\mathcal{S}} % permutation S
\newcommand{\scal}{\!\boldsymbol{\cdot}\!} % scalar product
 \newcommand{\mL}{\mathcal{L}} % angular momentum
 \newcommand{\mJ}{ {L}} % generalized angular momentum

\newcommand{\mT}{\mathcal{T}} % kinetic energy
  
\newcommand{\mU}{\mathcal{U}} % potential
 
\newcommand{\mA}{\mathcal{A}} % potential

\newcommand{\mB}{\mathcal{B}} % potential

\newcommand{\ff}{ {f}} % conformal factor
 
\newcommand{\bq}{ \mathbf{q}} % bold q  

\def\>#1{{\mathbf#1}}   

\newcommand{\bE}{ \mathbf{E}} % bold E
\newcommand{\bS}{ \mathbf{S}} % bold S
\newcommand{\bH}{ \mathbf{H}} % bold H
 \newcommand{\kk}{K}% coupling constant K
 
  \newcommand{\Phii}{\Psi}% group action

  \usepackage{float}
 \usepackage{placeins}

\begin{document}

\
  \vskip0.5cm

\noindent
\begin{center}
 {\Large \bf Nonlinear Lie--Hamilton systems: $t$-Dependent curved \\[6pt] oscillators and Kepler--Coulomb  Hamiltonians}
\end{center}

 \vskip 0.5cm

\centerline{{\sc Rutwig Campoamor-Stursberg$^{1,2}$, Francisco J. Herranz$^{3}$ and Javier de Lucas$^{4,5}$}}
\vskip 0.5cm

 \noindent
$^1$ Instituto de Matem\'atica Interdisciplinar, Universidad Complutense de Madrid, E-28040 Madrid,  Spain

\noindent
$^2$ Departamento de \'Algebra, Geometr\'{\i}a y Topolog\'{\i}a,  Facultad de Ciencias
Matem\'aticas, Universidad Complutense de Madrid, Plaza de Ciencias 3, E-28040 Madrid, Spain

\noindent
 $^3$ Departamento de F\'isica, Universidad de Burgos,
E-09001 Burgos, Spain 
 
\noindent
 $^4$ UW Institute for Advanced Studies, 
 Department of Mathematical Methods in Physics, University of Warsaw,
  ul. Pasteura 5, 02-093, Warsaw, Poland
  
  \noindent
 $^5$  Centre de Recherches Math\'ematiques, Universit\'e de Montr\'eal, 
 Succ. Centre-Ville, CP 6128, Montr\'eal, QC H3C 3J7, Canada

 \medskip

\noindent  E-mail: { 
 \href{mailto:rutwig@ucm.es}{\texttt{rutwig@ucm.es}}, \href{mailto:fjherranz@ubu.es}{\texttt{fjherranz@ubu.es}}, \href{mailto:javier.de.lucas@fuw.edu.pl}{\texttt{javier.de.lucas@fuw.edu.pl}}
 }
 
\medskip

  \begin{abstract}
\noindent
The Lie--Hamilton approach for $t$-dependent Hamiltonians is extended to cover the so-called nonlinear Lie--Hamilton systems, which are no longer related to a linear $t$-dependent combination of a basis of a finite-dimensional Lie algebra of functions $\mathcal{W}$, but an arbitrary  $t$-dependent function on $\mathcal{W}$. This novel formalism is accomplished through a detailed analysis of related structures, such as momentum maps and  generalized distributions, together with the extension of  the Poisson coalgebra method to a $t$-dependent frame, in order to systematize the construction of constants of the motion for nonlinear systems. Several relevant relations between nonlinear Lie--Hamilton systems, Lie--Hamilton systems, and collective Hamiltonians are analyzed.  The new notions and tools are illustrated with the study of the harmonic oscillator, H\'enon--Heiles systems and  Painlev\'e trascendents within a  $t$-dependent framework. In addition, the formalism is carefully applied to construct oscillators with a $t$-dependent frequency and Kepler--Coulomb systems with a $t$-dependent coupling constant on the $n$-dimensional sphere, Euclidean and hyperbolic spaces, as well as on some spaces of non-constant curvature.
 \end{abstract}

  \medskip
\medskip

\medskip


\smallskip
\noindent
\textbf{PACS 2010 codes}:    {02.30.Hq, 02.20.Sv, 45.10.Na,  45.20.Jj}

\newpage
  
  \tableofcontents

  \newpage
  
  %%%%%%%%%%%%%%%%%%%%%%%%%%%%%%%%%%%%%%%%%%%
\section{Introduction}

Lie--Hamilton systems constitute an interesting class of Lie systems, i.e., $t$-dependent systems of first-order ordinary differential equations admitting a (nonlinear) superposition principle \cite{LS93,Ueno72}, as they combine the algebraic properties with the compatibility with a Poisson structure, a fact that enriches considerably their interpretation and allows the application of additional techniques and tools for their analysis. The theory of Lie--Hamilton systems (LH in short) formally appeared around ten years ago in \cite{CLS13}, although several types of LH systems had already been considered in previous work (see e.g.~\cite{CGM00,CR03}). However, the formalism and the fundamental characteristics of LH systems were first given in \cite{CLS13}, also stressing their potential applications and relevance in the description of physical phenomena. Although the property of being a LH system is rather an exception than a general rule in the realm of $t$-dependent Hamiltonian systems, their significance is due to their applications. For instance, Smorodinsky--Winternitz oscillators, certain higher-order Riccati equations in Hamiltonian form, some trigonometric Hamiltonian systems in classical mechanics, epidemic models, and several other Hamiltonian systems appearing in physics can be studied by means of the LH formalism \cite{LS20}. A significant boost to the study of LH systems emerged from the results in \cite{BCHLS13}, where it was shown that LH systems admit a superposition rule that can be obtained  through the so-called coalgebra method \cite{BBR06,BBHMR09} in a more practical manner than previous approaches. Moreover,  superposition rules can be interpreted both geometrically and algebraically (see \cite{LS20} and references therein). For instance, the well-known superposition rule for Riccati equations was proved to be a consequence of a Casimir element of the Lie algebra $\mathfrak{sl}(2,\mathbb{R})$. Another influential work was the local classification of finite-dimensional Lie algebras of Hamiltonian vector fields on the plane and the analysis of their superposition rules and applications~\cite{BBHLS15}, which was based on the previous classification of Lie systems on   $\mathbb{R}^{2}$ due to Lie and analyzed in detail, much later, in \cite{GonzalezLopez1992}.  

These results motivated many applications and new results \cite{LS20,CFH23,HLT17,C139,CCH24,CCH25} (see also references therein). Moreover, the displayed methods allowed generalizations to other types of Lie systems with compatible geometric structures, such as Dirac--Lie or Jacobi--Lie systems  \cite{LS20,AHKR_20,AHR_22}. The Poisson coalgebra deformation theory using quantum algebra structures was successfully applied to LH systems to obtain much more general Hamiltonian systems of physical interest \cite{BCFHL21}.

The present work provides a new, further reaching generalization of LH systems: the so-called {\em nonlinear Lie--Hamilton systems}. In short, a nonlinear LH system is a $t$-dependent Hamiltonian system on a Poisson manifold $(\mm,\Lambda)$ related to a $t$-dependent  Hamiltonian function $h(t,x)=F(t,h_1,\ldots,h_r)$, where $h_1,\ldots,h_r$ span a finite-dimensional Lie algebra of functions relative to the Poisson bracket of the Poisson manifold and $F$ is any $t$-dependent function. Remarkably, we find that these dynamical systems can be understood as a $t$-dependent function of the form $h(t,x)=F(t,J(x))$ for an appropriate function $F\in C^\infty(\mathbb{R}\times \mathfrak{g}^*)$ and a momentum map $J:\mm\rightarrow \mathfrak{g}^*$.

This approach extends the results on the theory of applications of the Poisson coalgebra symmetry \cite{BBR06,BBHMR09} to a $t$-dependent context. 
Moreover, we unveil relations of these new systems with very relevant types of integrable systems appearing in the literature that seem to have passed unnoticed. For instance, we show that LH systems are $t$-dependent extensions of the so-called {\em collective Hamiltonians}, a concept that has attracted a lot of attention due to its wide physical applications \cite{GS80,GS84}. 
It should be observed that both theories have been  developed independently, albeit their striking similarities. This connection proves that nonlinear LH systems can be used to study rigid rotors with a $t$-dependent inertia tensor, liquid drops with a $t$-dependent configuration, and  nuclear collective models with $t$-dependent parameters, to merely mention a few possibilities (cf.~\cite{GS80}).

The standard approach to momentum maps starts by defining a Hamiltonian Lie group action, and subsequently finding a momentum map compatible with it. Here, given a finite-dimensional Lie algebra of Hamiltonian functions on a Poisson manifold, one can consider a basis of them as the coordinates of a momentum map associated with a local Lie group action. This is immediate when the Lie algebra of fundamental vector fields of the Hamiltonian functions is isomorphic to the Lie algebra of Hamiltonian functions, but requires a rather technical proof if there is no Lie algebra isomorphism.

The structure of the paper is as follows. After recalling in the next section the fundamental properties of Lie and LH systems, we introduce the notion of nonlinear LH systems through the construction of novel $t$-dependent H\'enon--Heiles Hamiltonians in Section~\ref{s3}. Next, we set the basic mathematical theory of nonlinear LH systems. In particular, its relation to momentum maps is established in Section~\ref{s4}, whereas fundamental geometrical structures are developed in Section~\ref{s5}, where  it is proved that the evolution of such nonlinear systems is restricted to certain generalized distributions. We also show that every $t$-dependent Hamiltonian system related to a Poisson bivector can be considered as a nonlinear LH system at a generic point. These properties are illustrated by means of Painlev\'e trascendents in Section~\ref{s51}.

Some generic results about $t$-dependent and $t$-independent constants of the motion for nonlinear LH systems are then presented in Section~\ref{s6}, where we also  show that properties of nonlinear LH systems are related to $t$-dependent Hamiltonian systems on the dual of a Lie algebra relative to its natural Kostant--Kirillov--Souriau  Poisson bivector. The extension of the integrability coalgebra method to a $t$-dependent frame is addressed in detail in Section~\ref{s7} which, for further applications, is particularized to the case of Poisson $\mathfrak{sl}(2,\mathbb{R})$-coalgebra symmetry in Section~\ref{s8}. It is worth noting that, as a relevant result, we  construct generic Hamiltonian systems on $n$-dimensional spherically symmetric spaces, which are formed by the superposition of  a $t$-dependent central potential with $n$ arbitrary $t$-independent Rosochatius--Winternitz (or ``centrifugal") potentials. Among  the  great number of $t$-dependent central potentials, in Section~\ref{s9} we focus on isotropic oscillators with a $t$-dependent frequency and Kepler--Coulomb Hamiltonians  with a $t$-dependent coupling constant. Such $t$-dependent nonlinear 
LH systems are explicitly obtained in the three classical Riemannian spaces of constant curvature; the sphere, the Euclidean and the hyperbolic space, as well as in some  spaces of non-constant curvature. To the best of our knowledge, all these systems constitute novel results. Finally, Section \ref{Sec:Conc} summarizes the work and comments on further future prospects and research directions.

 %%%%%%%%%%%%%%%%%%%%%%%%%%%%%%%%%%%%%%%%%%%

\section{Lie systems and Lie--Hamilton systems}
\label{s2}
 
We first set some general assumptions for the rest of this work. Unless otherwise stated, all structures are assumed to be smooth and globally defined. Moreover,  $\mm$ is assumed to be an $n$-dimensional manifold, which is further considered to be connected. When speaking about Hamiltonian systems, we always refer to $t$-dependent Hamiltonian systems relative to a Poisson manifold $(\mm,\Lambda)$ with Poisson bivector $\Lambda$. These simplifications highlight the main parts of our presentation and allow us to skip minor technical details.

We shortly review the theory of Lie and LH systems (see \cite{CGM00, LS20, Dissertations, PW83} for further details). The standard approach to Lie systems is through the theory of $t$-dependent vector fields \cite{Dissertations}. 

Any $t$-dependent system of first-order differential equations on  $\mm$ of the form
\begin{equation}\label{Eq:Sys}
\frac{\dd x^i}{\dd t}=X^i(t,x),\qquad \forall t\in \mathbb{R},\qquad \forall x=\{x^1,\ldots,x^n\}\in \mm, \qquad i=1,\ldots,n,
\end{equation}
can be associated with a {\em $t$-dependent vector field} on $\mm$, namely a map $X:\mathbb{R}\times \mm\rightarrow {\rm T}\mm$ given by
\begin{equation}\label{Eq:VF}
	X(t,x)=\sum_{i=1}^nX^i(t,x)\frac{\partial}{\partial x^i},\qquad \forall t\in \mathbb{R},\qquad \forall x\in \mm,
\end{equation}
which is equivalent to a $t$-parameterized family of vector fields $X_t:x\in \mm\mapsto X(t,x)\in {\rm T}\mm$, with $t\in \mathbb{R}$, and vice versa. This suggests to identify the $t$-dependent vector field $X$ given by (\ref{Eq:VF}) with its associated $t$-dependent system of differential equations (\ref{Eq:Sys}). This identification can be made without leading to any misunderstanding (cf.~\cite{LS20,Dissertations}).

If $A$ is a Lie algebra with Lie bracket $[\cdot,\cdot]_A:A\times A\rightarrow A$ and $B$ is a subset of $A$, we write ${\rm Lie}\bigl(B,[\cdot,\cdot]_A\bigr)$ for the smallest Lie subalgebra of $A$ containing $B$. It follows that ${\rm Lie}\bigl(B,[\cdot,\cdot]_A\bigr)$ is spanned by $A$ and the linear combinations of successive Lie brackets of elements of $A$, namely
$$
v_1,\ [v_1,v_2]_A,\ \bigl[v_1,[v_2,v_3]_A\bigr]_A,\ \bigl[[v_1,v_2]_A,[v_3,v_4]_A\bigr]_A,\ldots\qquad \forall v_1,v_2,v_3,v_4,\ldots \in A.
$$
Every $t$-dependent vector field $X$ on $\mm$ gives rise to a Lie algebra of vector fields, $V^X$, on $\mm$ characterized univocally by the fact that it is the smallest Lie algebra of vector fields (in the sense of inclusion) containing the vector fields $\{X_t\}_{t\in \mathbb{R}}$. In particular, $V^X={\rm Lie}\bigl(\{X_t\}_{t\in \mathbb{R}},[\cdot,\cdot] \bigr)$, where $[\cdot,\cdot]$ stands for the standard vector fields commutator.
\begin{definition} 
\label{def21}
A {\em superposition rule} for a system $X$ on a manifold $\mm$ is a function $\Phii:\mm^{m}\times \mm\rightarrow
\mm$ of the form $x=\Phii\bigl(x_{(1)}, \ldots,x_{(m)};k\bigr)$ bringing the general
 solution $x(t)$ of $X$ into the form  $x(t)=\Phii \bigl(x_{(1)}(t), \ldots,x_{(m)}(t);k\bigr),$
for any generic family $  x_{(1)}(t),\ldots,x_{(m)}(t) $ of
particular solutions of $X$ and $k\in \mm$. {\em Lie systems} are systems of first-order differential equations admitting a superposition rule (see \cite{Dissertations,PW83,CGM07} for details). 
 \end{definition}

For instance, $t$-dependent first-order systems of linear ordinary differential equations admit  linear superposition rules. In this case, such a system  is defined on a certain $\mathbb{R}^n$, the   superposition principle takes the form $\Phii\bigl(x_{(1)},\ldots, x_{(n)},k_1,\ldots,k_n\bigr)=\sum_{i=1}^nk_i x_{(i)}(t)$, the particular solutions $   x_{(1)}(t),\ldots,x_{(n)}(t)  $ must be linearly independent, and $k=\{k_1,\ldots,k_n\}$ is related to the initial conditions of the system $X$ under study.  The requirements ensuring that the system associated with a $t$-dependent vector field possesses a superposition rule are
described by the {\em Lie Theorem} \cite{LS93,Dissertations,CGM07}.

\begin{theorem}{\bf (Lie Theorem)}
 A system $X$ on $\mm$ admits a superposition rule if and only if 
$$
X(t,x)=\sum_{\alpha=1}^rb_\alpha(t)X_\alpha(x),\qquad \forall t\in \mathbb{R}, \qquad \forall x\in \mm,
$$
for a family $b_1(t),\ldots,b_r(t) $  of $t$-dependent functions and a set of vector fields $   X_1,\ldots,X_r $ on $\mm$ spanning an $r$-dimensional  Lie algebra of vector fields $V$,
referred to as  Vessiot--Guldberg Lie algebra associated with
$X$~\cite{LS93,PW83}.
\end{theorem}

In brief, the Lie Theorem states that  $X$ is a Lie system if and only if $V^X$ is finite-dimensional. To illustrate the Lie Theorem and superposition rules,  let us consider  a {\em Riccati equation} \cite{In56,Anderson,Eg07,Wi08}, namely the $t$-dependent differential equation on $\mathbb{R}$ given by
\begin{equation}\label{Ric}
\frac{{\rm d} x}{{\rm d} t}=b_1(t)+b_2(t)x+b_3(t)x^2,\qquad \forall t\in \mathbb{R},\qquad  \forall x\in \mathbb{R},
\end{equation}
where $b_1(t),b_2(t),b_3(t)$ are arbitrary $t$-dependent functions. Note that it is related to the $t$-dependent vector field $X=b_1(t)X_1+b_2(t)X_2+b_3(t)X_3$, where
\[
X_1=\frac{\partial}{\partial x},\qquad X_2=x\frac{\partial}{\partial x},\qquad X_3=x^2\frac{\partial}{\partial x}
\]
span a Lie algebra of vector fields, $V$, isomorphic to $\mathfrak{sl}(2,\mathbb{R})$ \cite{PW83}.  Hence, $V$ determines a Vessiot--Guldberg Lie algebra  for (\ref{Ric}), which becomes a Lie system. It is worth noting that  Riccati equations were classically assumed to satisfy that $b_1(t)b_3(t)\neq 0$, but, in the context of Lie systems, no restrictions on the coefficients in (\ref{Ric}) are imposed~\cite{LS20,Dissertations}.
 The general solution of (\ref{Ric}) can be brought into the form \cite{PW83}
$$
x(t)=\frac{x_{(1)}(t)\bigl(x_{(3)}(t)-x_{(2)}(t) \bigr)+k\,x_{(3)}(t)\bigl(x_{(1)}(t)-x_{(2)}(t)\bigr)}{ x_{(3)}(t)-x_{(2)}(t) +k\bigl(x_{(1)}(t)-x_{(2)}(t)\bigr)},
$$
where $ x_{(1)}(t),x_{(2)}(t),x_{(3)}(t) $ are three different particular solutions of the equation  (\ref{Ric}) and $k$ is an arbitrary constant. Hence, such a  relation is a superposition rule for the Riccati equation $\Phii:\bigl(x_{(1)},x_{(2)},x_{(3)};k\bigr)\in\mathbb{R}^3\times\mathbb{R}\mapsto
x\in\mathbb{R}$.
 
There exist several ways to derive superposition rules for Lie systems
\cite{Dissertations,PW83}. We now survey the distributional method
 devised in \cite{Dissertations,CGM07}, which is based on {\em diagonal
prolongations}. The reason of using this approach is that it allows for the use of geometric structures for determining and studying superposition rules more easily than with other Lie group methods (see \cite{LS20,PW83} for details). Although there exists a nice geometric, coordinate-free method to characterize diagonal prolongations of vector fields \cite{Dissertations}, we will hereafter give a simpler, more intuitive definition. 
\begin{definition} Given a $t$-dependent vector field $X$ on $\mm$ with local expression
$$
X(t,x)=\sum_{i=1}^nX^i(t,x)\frac{\partial}{\partial x^i},
$$
its  {\em
diagonal prolongation} to $\mm^{m+1}$ is the  $t$-dependent vector field
$\widetilde X$ on $\mm^{m+1}$ of the form
$$
\widetilde{X}\bigl(t,x_{(1)},\ldots,x_{(m+1)} \bigr)=\sum_{i=1}^n\sum_{\ell=1}^{m+1}X^i\bigl(t,x_{(\ell)}\bigr)\frac{\partial}{\partial x^i_{(\ell)}},
$$
where   $x_{(\ell)}=\bigl\{x_{(\ell)}^1,\ldots,x_{(\ell)}^n \bigr\}$ are the coordinates induced by  $x=\{x^1,\ldots,x^n\}\in \mm$ assuming that   they are defined on each $\ell$-th copy of $\mm$ within $\mm^{m+1}$. 
\label{def23}
\end{definition}
Diagonal prolongations can be applied to standard vector fields on a manifold $\mm$   by simply skipping the $t$-dependence in the above definition, which  allows us to determine a superposition rule for a Lie system $X$ by using the following procedure (see \cite{Dissertations} for a detailed analysis of the method).

\begin{itemize}
 \item Take a basis $ \{ X_1,\ldots,X_r  \}$ of a Vessiot--Guldberg Lie algebra $V$ for $X$ and determine the minimum natural number $m$ so that the diagonal prolongations $ \widetilde{X}_1,\ldots, \widetilde{X}_r$ of $ X_1,\ldots,X_r$ to $\mm^m$, respectively, become linearly independent at a generic point of $\mm^m$.

\item Choose a coordinate system $x=\{x^1,\ldots,x^n\}$ on $\mm$ and consider that it is defined on each copy of $\mm$ in $\mm^{m+1}$. This allows for defining a coordinate system $\bigl\{x_{(1)},   \ldots,x_{(m+1)}\bigr\}$ on $\mm^{m+1}$.

\item  Obtain $n$ functionally independent first integrals $  \I_1,\ldots,\I_n  $ common to the diagonal prolongations $\bigl\langle \widetilde X_1,\ldots,\widetilde
X_r \bigr\rangle$, so that 
\begin{equation}\label{Eq:Con}
\frac{
\partial(\I_1,\ldots,\I_n)}{\partial \bigl(x^1_{(m+1)},\ldots,x^n_{(m+1)} \bigr) }\neq 0.
\end{equation}

\item If one sets $\I_i=k_i$, with $k_i\in \mathbb{R}$ and $i=1,\ldots,n$, condition (\ref{Eq:Con}) enables us  to determine the values of $x_{(m+1)}=\bigl\{x^1_{(m+1)},\ldots,x^n_{(m+1)}\bigr\}$ in terms of the remaining variables $\bigl\{  x_{(1)}, \ldots,x_{(m)}\bigr\}  $ and $  \{k_1,\ldots,k_n\}   $.
\end{itemize}
The obtained  expressions provide a superposition rule in terms of a generic family of $m$ particular solutions and $n$ constants, namely $ k_1,\ldots,k_n $. The previous method relies on determining common  and functionally independent  first integrals for
$\bigl\langle \widetilde X_1,\ldots,\widetilde X_r\bigr\rangle$ admitting a mild condition (\ref{Eq:Con}), which {\em a priori} requires solving a
system of PDEs. This amounts to determining certain $t$-independent constants of the motion for $\widetilde X$, which can be a quite laborious task~\cite{LS20,CGM07}. As shown in the literature,
this can be simplified in the case of LH systems \cite{LS20}, as these are defined on a Poisson manifold that allows for the obtainment of  $ \I_1,\ldots,\I_n $ in a  systematic algebraic/geometric manner, for instance, via the Poisson coalgebra method \cite{BBR06,BBHMR09}. This, among other reasons, justifies the following definition.

\begin{definition} 
\label{defLH}
A system $X$ on $\mm$ is said to be a {\em Lie--Hamilton system} if
$\mm$ admits a Poisson bracket $\{\cdot,\cdot\}:C^\infty(\mm)\times C^\infty(\mm)\rightarrow C^\infty(\mm)$  so that $V^X$ becomes a finite-dimensional Lie algebra of Hamiltonian vector fields relative to the Poisson manifold.
\end{definition}

Recall that a Poisson bracket $\{\cdot,\cdot\}:C^\infty(\mm)\times C^\infty(\mm)\rightarrow C^\infty(\mm)$ gives rise to a unique bivector field $\Lambda$ on $\mm$ satisfying the constraint $[\Lambda,\Lambda]_{\rm SN}=0$ relative to the Schouten--Nijenhuis bracket $[\cdot,\cdot]_{\rm SN}$. We call $\Lambda$  a so-called {\em Poisson bivector}. Conversely, a Poisson bivector $\Lambda$ on  a manifold $\mm$ gives rise to a Poisson bracket $\{\cdot\,\cdot\}_\Lambda:C^\infty(\mm)\times C^\infty(\mm)\rightarrow C^\infty(\mm)$ given by
$$
\{f,g\}_{\Lambda}=\Lambda(\dd f,\dd g),\qquad \forall f,g\in C^\infty(\mm).
$$ We refer to \cite{AM78} for technical details.
It is worth recalling that $\Lambda$ gives rise to a morphism $\widehat{\Lambda}:{\rm T}^*\mm\rightarrow {\rm T} \mm$ of the form $\widehat{\Lambda}(\alpha)(\beta)=\Lambda(\alpha,\beta)$ for arbitrary covectors $\alpha,\beta$ at any $x\in \mm$. Recall that $\Lambda$ is said to be {\em regular} at a point $x\in \mm$ if the rank of the restriction $\widehat{\Lambda}_{x'}:{\rm T}^*_{x'}\mm\mapsto {\rm T}_{x'}\mm$ is constant for $x'$ in an open neighbourhood of $x\in \mm$.

%%%%%%%%%%%%%%%%%%%%%%%%%%%%%%%%%%%%

\subsection{{\em t}-Dependent Smorodinsky--Winternitz Hamiltonian}
\label{s21}

To illustrate the definition of a LH system, let us analyse the Hamilton equations for the $t$-dependent  counterpart of the   $n$-dimensional  Smorodinsky--Winternitz (SW) system  \cite{WSUF65} (see also~\cite{WSUF67,Evans90,Evans91,GPS06}) defined on ${\rm T}^*\mathbb{R}^n_0:=\{(\>q,\>p)=(q_1,\dots,q_n,p_1,\dots, p_n)\in {\rm T}^*\mathbb{R}^n:   \prod_{i=1}^n|q_i|^{|c_i|}\neq 0 \}$ in the form~\cite{LS20}
\begin{equation}\label{LieS}
\left\{\begin{aligned}
\frac{\dd q_i}{\dd t}&=p_i,\\
\frac{\dd p_i}{\dd t}&=-\omega^2(t)q_i+\frac{c_i}{q_i^3},
\end{aligned}\right.  \qquad i=1,\ldots,n, 
\end{equation}
and arbitrary $c_1,\ldots,c_n\in \mathbb{R}$. 
This system corresponds to the mix of an isotropic harmonic oscillator with 
  $t$-dependent arbitrary frequency $\omega(t)$  and  $n$  Rosochatius (or Winternitz) potentials depending on the real arbitrary constants $c_i$, each of them yielding  a ``centrifugal" barrier. Therefore,
  if  all $c_i=0$,  it reduces to the $t$-dependent isotropic oscillator. We recall that for $n=1$ and with the identifications $q_1\equiv x$, $p_1\equiv y=\dd x/\dd t$ and $c_1=c$, the $t$-dependent SW system    (\ref{LieS}) on  ${\rm T}^*\mathbb{R}_0$ reproduces  the  \textit{Ermakov equation}~\cite{Ermakov1880}  (see also~\cite{Leach1991,Maamache1995,LS20,Leach2008}), which as a second-order ordinary differential equation  reads as 
 \begin{equation}
 \frac{{\dd^2 x}} {\dd t^2} = - \omega^{2}(t) x + \frac{c}{x^{3}},  
 \label{Ea}
 \end{equation}
  also known as the {\em Milne--Pinney  equation}~\cite{Milne1930,Pinney1950}.

 The system (\ref{LieS}) describes the integral curves of the $t$-dependent vector
field
\be
X_{\rm SW}(t,\>q,\>p)=\sum_{i=1}^n\left[p_i\frac{\partial}{\partial
q_i}-\left(\omega^2(t)q_i-\frac{c_i}{q_i^3}\right)\frac{\partial}{\partial
p_i}\right]
\label{xsw}
\ee 
on ${\rm T}^*\mathbb{R}^{n}_0$. Note that ${\rm T}^*\mathbb{R}^{n}_0$ admits a natural Poisson bivector 
\be
\Lambda=\sum_{i=1}^n\frac{\partial}{\partial q_i}\wedge\frac{\partial}{\partial p_i }
\label{bivector}
\ee
 related to the restriction to ${\rm T}^*\mathbb{R}^n_0$ of the canonical symplectic structure on ${\rm T}^*\mathbb{R}^n$.  Then, the $t$-dependent vector field (\ref{xsw})
 can be expressed as $X_{\rm SW}(t)=X_3+\omega^2(t)X_1$, for every $t\in \mathbb{R}$, with   
   vector fields  given by
\begin{equation}
X_1=-\sum_{i=1}^n q_i\frac{\partial}{\partial p_i},\qquad 
X_2=\frac{1}{2}\sum_{i=1}^n\left(p_i\frac{\partial}{\partial
p_i}-q_i\frac{\partial}{\partial q_i}\right),\qquad
X_3=\sum_{i=1}^n\left(p_i\frac{\partial}{\partial
q_i}+\frac{c_i}{q_i^3}\frac{\partial}{\partial p_i}\right),
\label{SWxs}
\end{equation}
 which satisfy the commutation relations
\begin{equation} 
[X_1,X_2]=X_1,\qquad [X_1,X_3]=2X_2,\qquad  [X_2,X_3]=X_3,
\label{xSW}
\end{equation}
so that $X_2$ is necessary to close a Lie algebra.  It follows that (\ref{LieS}) is a Lie system related to a Vessiot--Guldberg Lie algebra $V_{\rm SW}=\langle X_1,X_2,X_3\rangle$ isomorphic to $\mathfrak{sl}(2,\mathbb{R})$. Notice that if  all $c_i=c$, then  $X_{\rm SW} $ (\ref{xsw})
provides  the diagonal prolongation of a  Ermakov system to ${\rm T}^*\mathbb{R}^n_0$.

In addition,  $V_{\rm SW}$ consists of Hamiltonian vector fields. In fact, the Poisson bivector (\ref{bivector}) gives rise to a morphism $\widehat{\Lambda}:{\rm T}^*{\rm T}^*\mathbb{R}^n_0\rightarrow {\rm T} {\rm T}^*\mathbb{R}^n_0$ such that $X_\alpha=-\widehat\Lambda(\dd h_\alpha)$ $(\alpha=1,2,3)$ yielding
\begin{equation} 
h_1=\frac{1}{2}\sum_{i=1}^{n}{q_i^2},\qquad h_2=-\frac
12\sum_{i=1}^{n}{q_ip_i},\qquad
h_3=\frac{1}{2}\sum_{i=1}^{n}{\left(p_i^2+\frac{c_i}{q_i^2}\right)},
\label{hSW}
\end{equation}
which  fulfil  the Poisson brackets
\begin{equation} 
\{h_1,h_2\}_\Lambda=-h_1,\qquad \{h_1,h_3\}_\Lambda=-2h_2,\qquad 
\{h_2,h_3\}_\Lambda=-h_3. 
\label{commSW}
\end{equation}
Note that the   vector fields and Hamiltonian vector fields verify the usual conditions for the Lie derivative and inner product 
 \be
{\cal L}_{ {X}_\alpha}\omega =0 ,\qquad  \iota_{X_{\alpha}} \omega = \dd h_{\alpha},
\label{conditions}
\ee
with respect to the canonical 
symplectic form 
\be
\omega=\sum_{i=1}^n \dd q_i \wedge \dd p_i .
\nonumber
\ee

Hence, the Lie algebra $V_{\rm SW}$ consists of Hamiltonian vector fields relative to $\Lambda$ and the set of equations (\ref{LieS}) becomes a LH system with   $t$-dependent Hamiltonian for $X_{\rm SW}$ (\ref{xsw}) given by
\be
h_{\rm SW}(t,\>q,\>p)=h_3+\omega^2(t)h_1 = \frac{1}{2}\sum_{i=1}^{n}p_i^2+ \frac12  \sum_{i=1}^{n}\left( \omega^2(t) q_i^2 +\frac{c_i}{q_i^2}  \right) .
\label{hhSW}
\ee
 Observe that the Hamiltonian function $h_3$ determines both the kinetic energy and the Rosochatius--Winternitz potentials, while  $h_1$ provides the $t$-dependent isotropic oscillator potential, and the system (\ref{LieS}) thus comes from the Hamilton equations of $h_{\rm SW}$. 

 Furthermore, such Hamiltonian gives rise to what is called a Lie--Hamiltonian structure  and $\langle h_1,h_2,h_3 \rangle$ lead to a LH algebra. More specifically, one has the following definition.
\begin{definition}
\label{defLHS} A {\em Lie--Hamiltonian structure} is a triple  
$(\mm,\Lambda,h)$, where $(\mm,\Lambda)$ stands for  a Poisson manifold while $h:
(t,x)\in \mathbb{R}\times \mm\mapsto h_t(x)=h(t,x)\in  \mm$ is such that
 the space $\bigl(\ham= {\rm Lie}(\{h_t\}_{t\in\mathbb{R}},\{\cdot,\cdot\}_\Lambda),\{\cdot,\cdot\}_\Lambda \bigl)$ is  a finite-dimensional real Lie algebra. The space $\bigl(\ham,\{\cdot,\cdot\}_\Lambda \bigr)$ is called a {\em Lie--Hamilton   algebra}.
\end{definition}

Let us discuss some relevant facts on LH systems (see \cite{LS20} for details and proofs).  
\begin{theorem} 
\label{teorLHalgebra} A system $X$ on $\mm$ is a {\em LH system} if and only
if there exists a Lie--Hamiltonian structure $(\mm,\Lambda,h)$ such that $X_t=
-\widehat \Lambda( {\rm d} h_t)$ for every $t\in\mathbb{R}$. 
\end{theorem}

In the above theorem and definition, $\bigl(\ham,\{\cdot,\cdot\}_\Lambda \bigr)\equiv {\cal{H}}_\Lambda$ is called a {\em LH algebra} of 
$X$. Moreover, $\bigl({\rm T}^* \mathbb{R}^n_0,\Lambda,h_{\rm SW}\bigr)$ is a Lie--Hamiltonian structure for the LH system (\ref{LieS}) and $ h_1,h_2,h_3  $ span an associated LH algebra.

\begin{proposition}\label{IntLieHam0} A $t$-independent function is a constant of the
motion for a LH system $X$ if it Poisson commutes with all the elements of a certain LH algebra ${\cal{H}}_\Lambda$ of the system. If $X$ possesses a Vessiot--Guldberg Lie algebra of Hamiltonian vector fields with respect to a Poisson bivector $\Lambda$, the family $\mathcal{I}^X$ of its $t$-independent constants of the
motion form a Poisson algebra $\bigl(\mathcal{I}^X,\cdot,\{\cdot,\cdot\}_\Lambda\bigr)$.
\end{proposition}

%%%%%%%%%%%%%%%%%%%%%%%%%%%%%%%%%%%%%%%%%%%

\section{The notion of nonlinear Lie--Hamilton systems:    H\'enon--Heiles Hamiltonians}
\label{s3}

This section introduces nonlinear LH systems, which are proved to describe any $t$-dependent Hamiltonian system around a generic point. This will generalize significantly the techniques specific to LH systems to cover a much larger class of $t$-dependent Hamiltonian systems relative to a Poisson bivector that also retrieve, as a particular instance, $t$-independent Hamiltonian systems described by the coalgebra approach~\cite{BBR06,BBHMR09}. Indeed, these results can be understood as a generalization of both previous theories. It is worth noting that our formalism could be extended to many other more general geometric structures admitting Hamiltonian vector fields related to a Poisson algebra of functions. 

Before presenting the formal definition of nonlinear LH systems, let us first  introduce this concept together with its related structures by explicitly constructing the $t$-dependent version of the  well-known H\'enon--Heiles systems, which, to best of our knowledge, has not been considered before. 

 We recall that the original H\'enon--Heiles (HH) system  was introduced in~\cite{HH64} to model a Newtonian axially-symmetric galactic system. It is given by the following two-dimensional Hamiltonian
 $$
{H}=\dfrac{1}{2}\bigl(p_{1}^{2}+p_{2}^{2} \bigr) + \frac{1}{2}\bigl(q_{1}^{2}+q_{2}^{2}\bigr)+\lambda\left(
q_{1}^{2}q_{2}-\frac{1}{3} q_{2}^{3}\right),\qquad \lambda\in \mathbb{R},
\label{HHaut}
$$
where $\{q_1,q_2,p_1,p_2\}$ are the usual canonical variables with Poisson brackets $\{ q_i,p_j\}_\Lambda=\delta_{ij}$ for $i,j=1,2$ with respect to the  Poisson bivector $\Lambda$ (\ref{bivector}) for $n=2$. This system can be seen as  the combination of the isotropic harmonic oscillator with a cubic potential.
Remarkably, it is not  integrable in the sense of Liouville~\cite{Perelomov},  and provides a paradigmatic   nonlinear   system   exhibiting  chaotic behaviour. Nevertheless,  a multiparametric generalization  was subsequently considered in terms of  four real parameters $\Omega_1$,  $\Omega_2$, $\alpha$ and $\beta$  in the form
\begin{equation}
 {H}_{\rm HH}=\frac{1}{2}\bigl(p_{1}^{2}+p_{2}^{2}\bigr)+ \Omega_{1}  q_{1}^{2}+\Omega_{2} q_{2}^{2}+\alpha \bigl(
q_{1}^{2}q_{2}+\beta q_{2}^{3}\bigr) ,
\label{hhmulti}
\end{equation}
 which was  proven to be Liouville integrable only for {\em three} specific sets of values of the real parameters $\Omega_{1},\Omega_{2},\beta$, keeping $\alpha\ne 0$. The resulting distinguished families of integrable HH Hamiltonians have been extensively studied from different approaches~\cite{BSV,CTW,GDP,HietarintaRapid,Fordy83,Wojc,RGG,Conte,HonePLA,BB10,WU,She,Ren}  (see also  references therein), including connections with hierarchies of differential equations,  and are known as:
\begin{itemize}

\item The Sawada--Kotera (SK) Hamiltonian for $\beta=1/3$ and $\Omega_{2}=\Omega_{1}$. It has a constant of the motion quadratic in the momenta.  

\item  The Korteweg--de Vries (KdV) Hamiltonian when $\beta=2$  keeping $\Omega_{1}$ and $\Omega_{2}$ arbitrary.  The constant of the motion is also quadratic.

\item The Kaup--Kupershmidt (KK) Hamiltonian if $\beta=16/3$ and $\Omega_{2}=16\,\Omega_{1}$.  The constant of  the motion is quartic in the momenta.

\end{itemize}

These systems admit a direct interpretation in terms of integrable cubic perturbations of anisotropic oscillators, by introducing two frequencies $\omega_1^2=2\Omega_1$ and $\omega_2^2=2\Omega_2$ in ${H}_{\rm HH}$ (\ref{hhmulti}). 
In addition, we recall that the particular KdV system corresponding to the case $ \Omega_2=4  \Omega_1$, denoted ${\rm KdV}^{1:2}$,  entails  the possibility to add the so-called Ramani--Dorizzi--Grammaticos series of integrable homogeneous polynomial potentials~\cite{RDGprl}. Consequently, under this interpretation, the  SK, ${\rm KdV}^{1:2}$ and KK systems can be regarded as  integrable cubic perturbations of the $1:1$ (isotropic), $1:2$ and $1:4$ superintegrable  anisotropic  oscillators.

Now, we proceed to deduce the $t$-dependent generalization of ${H}_{\rm HH}$ and thus of the three integrable families. Consider a $t$-dependent Hamiltonian in ${\rm T}^*\mathbb{R}^2$ given by
\be
h_{\rm HH}(t,q_1,q_2,p_1,p_2)=\frac 12\bigl(p_1^2+p_2^2 \bigr)+V\bigl(t,q_1^2,q_2\bigr),
\label{HH}
\ee
for an arbitrary function $V:(t,z_1,z_2)\in \mathbb{R}^3\mapsto V(t,z_1,z_2)\in \mathbb{R}$.
This $t$-dependent Hamiltonian function gives rise to a $t$-dependent vector field $X$ on ${\rm T}^*\mathbb{R}^2$ such that each vector field $X_t$ is the Hamiltonian vector field of $(h_{\rm HH})_t$  with respect to the standard symplectic form on ${\rm T}^*\mathbb{R}^2$ for every $t\in \mathbb{R}$. Then, $X$ is related to the $t$-dependent system of differential equations
$$
\frac{\dd q_1}{\dd t}=p_1,\qquad
\frac{\dd q_2}{\dd t}=p_2,\qquad
\frac{\dd p_1}{\dd t}=-2q_1\frac{\partial V}{\partial z_1},\qquad
\frac{\dd p_2}{\dd t}=-\frac{\partial V}{\partial z_2}.
$$
Clearly, the $t$-independent version of $h_{\rm HH}$ (\ref{HH}) encompasses the multiparametric HH  Hamiltonian (\ref{hhmulti}). The functions
\be
h_1=p_1^2,\qquad h_2=q_1^2 ,\qquad h_3=q_1p_1,\qquad h_4=p_2,\qquad h_5=q_2,\qquad h_6=1,
 \label{hamHH}
 \ee
span a Lie algebra of functions, $\mathfrak{W}_{\rm HH}$, isomorphic to direct sum Lie algebra $\mathfrak{sl}(2,\mathbb{R})\oplus \mathfrak{h}_3$, where $\langle h_1,h_2,h_3\rangle \simeq \mathfrak{sl}(2,\mathbb{R})$, while $\langle h_4,h_5,h_6\rangle\simeq \mathfrak{h}_3$, with $\mathfrak{h}_3$ being the  Heisenberg Lie algebra. The non-vanishing Poisson brackets with respect to the Poisson bivector  (\ref{bivector}) (with $n=2$) read
 \be
  \{ h_1,h_2\}_{\Lambda}= - 4 h_3,\qquad  \{ h_1,h_3\}_{\Lambda}= -  2 h_1,\qquad \{ h_2,h_3\}_{\Lambda}=    2 h_2,\qquad\{ h_4,h_5\}_{\Lambda}= - h_6. 
\nonumber
 \ee
 The corresponding vector fields are derived from the condition (\ref{conditions}) and turn out to be
\begin{equation}
\begin{array}{lll}
\displaystyle{X_1 =2p_1\frac{\partial}{\partial q_1} }, &\quad \displaystyle{X_2=-2q_1\frac{\partial}{\partial p_1} }, &\quad  \displaystyle{ X_3=q_1\frac{\partial}{\partial q_1}-p_1\frac{\partial}{\partial p_1}}, \\[10pt]  
\displaystyle{ X_4 =\frac{\partial}{\partial q_2} }, & \quad  \displaystyle{ X_5=-\frac{\partial}{\partial p_2} }, & \quad   X_6=0,
\end{array}
\nonumber
\end{equation}
fulfilling the Lie brackets 
  \be
 [ X_3,X_1]=  - 2 X_1,\qquad [ X_3,X_2]=    2 X_2,\qquad [  X_1,X_2]=  4 X_3,\qquad [   X_4,X_5]= 0. 
 \nonumber
 \ee

  There exists a Lie group action of ${\rm SL}(2,\mathbb{R})\times {\rm H}_3$, where ${\rm H}_3$ is the Heisenberg group,  on ${\rm T}^*\mathbb{R}^2$ given by  
$$
\Phi\left(\left(\begin{array}{cc}A&B\\C&D\end{array}\right),\left(\begin{array}{ccc}1&a&c\\0&1&b\\0&0&1\end{array}\right);(q_1,p_1,q_2,p_2)\right)=(A q_1+B p_1,C q_1+D p_1,q_2+a,p_2-b),
$$
admitting  an associated mapping
$$
J:(q_1,q_2,p_1,p_2)\in {\rm T}^*\mathbb{R}^2\mapsto \sum_{\alpha=1}^6h_\alpha(q_1,q_2,p_1,p_2) e^\alpha\in \bigl(\mathfrak{sl}(2,\mathbb{R})\oplus \mathfrak{h}_3 \bigr)^*,
$$
where $ \{ e^1,\ldots,e^6  \}$ is a basis of $\bigl(\mathfrak{sl}(2,\mathbb{R})\oplus \mathfrak{h}_3\bigr)^*$ and the Hamiltonian functions $h_1,\ldots,h_6$  are given by  (\ref{hamHH}). Let  $\{\lambda_1,\ldots,\lambda_6 \}$ be the linear coordina\-tes related to the basis $ \{ e^1,\ldots,e^6  \} $ and  consider  the function $F$  on   $C^\infty \bigl( \mathbb{R}\times \bigl(\mathfrak{sl}(2,\mathbb{R})^*\oplus \mathfrak{h}_3^*\bigr) \bigr)$ given by
$$
F\bigl(t,\lambda_1 e^1+ \cdots + \lambda_6e^6 \bigr)=\frac 12\bigl(\lambda_1 +\lambda_4^2\bigr)+V(t,\lambda_2,\lambda_5).
$$
In these terms, the $t$-dependent HH Hamiltonian (\ref{HH}) can    be written as
$$
h_{\rm HH}(t,q_1,q_2,p_1,p_2)=F\bigl(t,J(q_1,q_2,p_1,p_2) \bigr),
$$
 which comprises the $t$-dependent version of the multiparametric HH Hamiltonian (\ref{hhmulti}) for a particular choice of the potential $V\bigl(t,q_1^2,q_2\bigr)$ in the form
  \begin{equation}
  \begin{split}
 h_{\rm HH}&=\frac{1}{2}\bigl( h_1+h_4^2\bigr)+ \Omega_{1}(t)  h_2+\Omega_{2}(t) h_5^2+\alpha(t) \bigl(
h_2 h_5+\beta(t) h_{5}^{3}\bigr) \\[2pt]
&=\frac{1}{2}\bigl(p_{1}^{2}+p_{2}^{2}\bigr)+ \Omega_{1}(t)  q_{1}^{2}+\Omega_{2}(t) q_{2}^{2}+\alpha(t) \bigl(
q_{1}^{2}q_{2}+\beta(t) q_{2}^{3}\bigr),
\end{split}
\label{hhmulti2}
\end{equation}
leading, via the Hamilton equations,  to the   $t$-dependent system of nonlinear coupled differential equations
 \begin{equation}
\left\{\begin{aligned}
\, \frac{\dd q_1}{\dd t}&=p_1,\qquad \frac{\dd p_1}{\dd t} =-2  \Omega_{1}(t)  q_1- 2\alpha(t) q_1 q_2  , \\[2pt]
\, \frac{\dd q_2}{\dd t}&=p_2,\qquad  
\frac{\dd p_2}{\dd t}=-2  \Omega_{2}(t)  q_2-  \alpha(t) \bigl( q_1^2 + 3 \beta(t) q_2^2 \bigr).
\end{aligned}\right. 
 \label{hhmulti3}
\end{equation}
   The $t$-dependent counterpart of the    SK, ${\rm KdV}^{1:2}$ and KK   integrable HH systems are shown in Table~\ref{table1}.

  %%%%%%%%%%%%%%%% TABLE 1%%%%%%%%%%%%%%%%%%
 
 \begin{table}[t!]
{\small
\caption{ \small{Three relevant $t$-dependent H\'enon--Heiles   Hamiltonians from the multiparametric system (\ref{hhmulti2}) with Hamilton equations (\ref{hhmulti3}). In all cases  $ p_i={\dd q_i}/{\dd t}$ for $i=1,2$.}}
  \begin{center}
\noindent 
\begin{tabular}{ l}
\hline

\hline
\\[-4pt]
$\bullet$ $t$-dependent Sawada--Kotera:  \  $\beta(t)=\frac 13 $,\quad $\Omega_{1}(t)=\Omega_{2}(t)=\frac 12 \omega^2(t)$. \\[4pt]
   $\displaystyle{\quad h_{\rm SK} =\frac{1}{2}\bigl(p_{1}^{2}+p_{2}^{2}\bigr)+ \frac{1}{2}\omega^2(t)\bigl(q_{1}^{2}+q_{2}^{2}\bigr) +\alpha(t) \biggl(
q_{1}^{2}q_{2}+\frac 13 q_{2}^{3}\biggr) }$
   \\[8pt]
 $\displaystyle{\quad   \frac{\dd p_1}{\dd t}= - \omega^2(t)q_1 - 2 \alpha(t) q_1 q_2 \qquad   \frac{\dd p_2}{\dd t}= -\omega^2(t)q_2 -  \alpha(t) \bigl( q_1^2+ q_2^2\bigr) }$
     \\[14pt]
$\bullet$ $t$-dependent  Korteweg--de Vries:  \ $\beta(t)=2$,\quad $ \Omega_{1}(t)=\frac 12\omega^2(t)$,\quad $ \Omega_{2}(t)=2\omega^2(t)$. \\[4pt]
   $\displaystyle{\quad h_{ {\rm KdV}^{1:2}} =\frac{1}{2}\bigl(p_{1}^{2}+p_{2}^{2}\bigr)+ \frac{1}{2}\omega^2(t)\bigl(q_{1}^{2}+4 q_{2}^{2}\bigr) +\alpha(t) \bigl(
q_{1}^{2}q_{2}+2 q_{2}^{3}\bigr) }$
   \\[8pt]
 $\displaystyle{\quad   \frac{\dd p_1}{\dd t}= - \omega^2(t)q_1 - 2 \alpha(t) q_1 q_2 \qquad   \frac{\dd p_2}{\dd t}= -4 \omega^2(t)q_2 -  \alpha(t) \bigl( q_1^2+ 6 q_2^2\bigr) }$
     \\[14pt]
$\bullet$ $t$-dependent  Kaup--Kupershmidt:  \  $\beta(t)=\frac{16}3$,\quad $ \Omega_{1}(t)=\frac 12\omega^2(t)$,\quad $ \Omega_{2}(t)=8\omega^2(t)$. \\[4pt]
   $\displaystyle{\quad h_{\rm KK} =\frac{1}{2}\bigl(p_{1}^{2}+p_{2}^{2}\bigr)+ \frac{1}{2}\omega^2(t)\bigl(q_{1}^{2}+16 q_{2}^{2}\bigr) +\alpha(t) \biggl(
q_{1}^{2}q_{2}+\frac {16}3 q_{2}^{3}\biggr) }$
   \\[8pt]
 $\displaystyle{\quad   \frac{\dd p_1}{\dd t}= - \omega^2(t)q_1 - 2 \alpha(t) q_1 q_2 \qquad   \frac{\dd p_2}{\dd t}= -16 \omega^2(t)q_2 -  \alpha(t) \bigl( q_1^2+ 16 q_2^2\bigr) }$
     \\[12pt]
 \hline

\hline
\end{tabular}
 \end{center}
\label{table1}
}
\end{table}

%%%%%%%%%%%%%%%%%%%%%%%%%%%%%%%%%%%

These results justify the following natural definition:

\begin{definition}
\label{Def:nonLH} 
A {\em nonlinear Lie--Hamilton system} is a four-tuple 
$$
\left(\mm,\Lambda,X,J:x\in \mm\mapsto \sum_{\alpha=1}^rh_\alpha(x) e^\alpha\in \mathfrak{g}^*\right)
$$ 
such that $\langle h_1,\ldots,h_r \rangle$ is an $r$-dimensional Lie algebra of Hamiltonian functions  $\mathfrak{W}$ relative to $\{\cdot,\cdot\}_\Lambda$ and $X$ admits a $t$-dependent Hamiltonian function of the form
$$
h(t,x)=F(t,J(x)),\qquad \forall t\in \mathbb{R},\qquad \forall x\in \mm,
$$
for some function  $F\in C^\infty(\mathbb{R}\times \mathfrak{g}^*)$, where $\mathfrak{g}$ is a Lie algebra isomorphic to $\mathfrak{W}$, and a basis $ \{ e^1,\ldots,e^r \}$ of $\mathfrak{g}^*$.
\end{definition}

It is worth noting that the function $F\in C^\infty(\mathbb{R}\times \mathfrak{g}^*)$ in Definition \ref{Def:nonLH} can be denoted, for a given basis $ \{ e^1,\ldots,e^r  \}$, in the following manner
$$
F(t,J(x))=F\bigl(t,h_1(x),\ldots,h_r(x) \bigr),\qquad \forall t\in \mathbb{R},\qquad \forall x\in \mm,
$$
where $h_\alpha(x)=\langle J(x),e_\alpha\rangle$ for $\alpha=1,\ldots,r$ and a dual basis $\{e_1,\ldots,e_r\}$ to $\{e^1,\ldots,e^r\}$.
Moreover, the Poisson bivector on $\mm$ allows us to relate $h$ to a $t$-dependent vector field on $\mm$ such that 
$$
X_t=-\widehat{\Lambda}(\dd h_t)=-\Lambda(\dd h_t,\cdot),\qquad \forall t\in \mathbb{R},
$$
where, for each $t\in \mathbb{R}$, one has that $h_t(x)=h(t,x)$, with $x\in \mm$, is a function on $\mm$. The $t$-dependent function on $\mm$ given by $h\in C^\infty(\mathbb{R}\times \mm)$ gives rise to a $t$-dependent Hamiltonian vector field on $\mm$, whose properties can be studied via $h$. It is worth noting that the vector fields of the functions $\mathfrak{W}$ may span a Lie algebra $V$ with $\dim V<\dim \mathfrak{W}$.

%%%%%%%%%%%%%%%%%%%%%%%%%%%%%%%%%%%%%%%%%%%

\section{Momentum maps}
\label{s4}

Definition \ref{Def:nonLH} is geometrically elegant and of theoretical interest, as it unveils several features of LH systems, integrable systems with a Poisson coalgebra symmetry, as well as collective Hamiltonians. To properly understand its geometric relevance, it is necessary to explain the meaning of the mapping $J$ appearing in the previous section and its geometric properties. This is the main objective of this section.

Consider a Poisson manifold $(\mm, \Lambda)$ consisting of a manifold $\mm$ and a Poisson bracket $\{\cdot,\cdot\}:C^\infty(\mm)\times C^\infty(\mm)\rightarrow C^\infty(\mm)$ associated with a Poisson bivector $\Lambda$, i.e.~a skew-symmetric two-contravariant tensor field such that $[\Lambda,\Lambda]_{\rm SN}=0$ for the Schouten--Nijenhuis bracket~\cite{Va_94}. Consider a Lie group action $\Phi:G\times \mm\rightarrow \mm$ whose fundamental vector fields are Hamiltonian relative to $\{\cdot,\cdot\}$. In particular, take a basis $\{e_1,\ldots,e_r\}$ of the Lie algebra $\mathfrak{g}$ of $G$ and define its fundamental vector fields by
$$
X_\alpha(x)=\frac{\dd}{\dd t}\bigg|_{t=0}\Phi\bigl(\exp(t \, e_\alpha),x \bigr),\qquad \alpha=1,\ldots,r,\qquad \forall x\in \mm.
$$
Let $  h_1,\ldots,h_r  $ be some Hamiltonian functions of  the vector fields $ X_1,\ldots,X_r  $, respectively. It can be proved that such Hamiltonian functions can be enlarged with several Casimir functions (related to zero Hamiltonian vector fields \cite{CGLS14}) to span an  $r'$-dimensional Lie algebra $\mathfrak{W}$ relative to the Poisson bracket $\{\cdot,\cdot\}_\Lambda$. Now, define a mapping $J:x\in \mm\mapsto \sum_{\alpha=1}^{r'}h_\alpha(x)e^\alpha\in  \mathfrak{g}^*$, where $ \{e^1,\ldots,e^{r'}  \}$ is the dual basis of   $\{ e_1,\ldots,e_{r'}\}$ (the latter closing the same commutation relations as $\{ h_1,\ldots, h_{r'}\}$). In other words, we have that $\langle J(x),e_\alpha\rangle=h_\alpha$ is a Hamiltonian function of $X_\alpha$ for $\alpha=1,\ldots,r$. We assume  that
$$
\bigl\{J_\xi,J_\mu\bigr\}_{\Lambda}=J_{[\xi,\mu]},\qquad \forall \xi,\mu\in \mathfrak{g}.
$$
This is why $J$ is called an {\em ${\rm Ad}$-equivariant momentum map}. This condition is not too restrictive and quite ubiquitous in the literature \cite{OR04}. Since $G$ is assumed to be connected, then (see \cite{AM78})
$$
J\circ \Phi(g,x)={\rm Ad}_{g^{-1}}^TJ(\Phi(g,x)),\qquad \forall g\in G,\qquad  \forall x\in \mm .
$$ 
  
It is relevant to stress that $\mathfrak{g}^*$ admits a Poisson bracket: the Kostant--Kirillov--Souriau (KKS in short) bracket. More specifically, consider the dual $\mathfrak{g}^*$ of a Lie algebra $\mathfrak{g}$ with a basis $\{ e_1,\ldots,e_{r'}\}$ such that $[e_\alpha,e_\beta]=\sum_{\gamma=1}^{r'}c_{\alpha\beta}^\gamma e_\gamma$ for $\alpha,\beta=1,\ldots,r'$. Then, the KKS bracket on $\mathfrak{g}^*$ is related to the Poisson bivector  
\be
\Lambda_{\mathfrak{g}^*} =\sum_{1\leq \alpha<\beta}^{r'} \sum_{\gamma=1}^{r'}c_{\alpha\beta}^\gamma e_\gamma\frac{\partial}{\partial e_\alpha}\wedge \frac{\partial}{\partial e_\beta},
\label{biv}
\ee 
where we understand the elements of $\mathfrak{g}$ as coordinates in the dual $\mathfrak{g}^*$ via the known isomorphism $\mathfrak{g}\simeq \mathfrak{g}^{**}$. Moreover, it can be proved that $J_*\Lambda=\Lambda_{\mathfrak{g}^*}|_{J(\mm)}$.

To show that every $J$ in Definition \ref{Def:nonLH} can be understood as a momentum map, at least locally, which is our main concern in this work, we prove the following statement. It allows us to relate a LH algebra to a momentum map whose coordinates give a basis of the LH algebra.

\begin{proposition}
\label{Prop:EqLieMom} 
Let $\mathfrak{W}$ be  a $p$-dimensional Lie algebra of Hamiltonian functions   relative to a Poisson bracket on $C^\infty(\mm)$, and let $\mathfrak{g}$ be the abstract Lie algebra isomorphic to $\mathfrak{W}$. Assume, without  loss of generality, a basis $\{h_1,\ldots,h_p\}$ of $\mathfrak{W}$ so that $X_{h_1}=\ldots=X_{h_s}=0$ and $X_{h_{s+1}},\ldots,X_{h_{s+r}}$ span an $r$-dimensional Lie algebra of vector fields. Choose a basis $\{ e_1,\ldots,e_p\} $ closing the same commutation relations as $\{ h_1,\ldots,h_p\}$. 
Then, there exists a Lie group action $\Phi:G\times \mm\rightarrow \mm$ admitting a momentum map $J:\mm
\rightarrow \mathfrak{g}^*$ such that $\langle J(x),e_\alpha\rangle=h_\alpha(x)$ for $\alpha=1,\ldots,p$.  
\end{proposition}

\begin{proof}The Hamiltonian functions $  h_1,\ldots,h_p $ admit Hamiltonian vector fields $ X_1,\ldots,X_p $, respectively. The map $\phi:h\in \mathfrak{W}
\mapsto -X_h\in V$ relating each element in $\mathfrak{W}$ into minus its associated Hamiltonian vector field is a Lie algebra surjective morphism and, in consequence, $V\simeq \mathfrak{W}/\ker \phi$. Let $\mathfrak{g}'$ be the abstract Lie algebra isomorphic to $V$. Then, there exists a Lie group\footnote{Strictly speaking, it is only a local Lie group action, but as we are focused on local aspects, the difference is not relevant and it will be skipped for simplicity.} action $\varphi:G'\times \mm\rightarrow \mm$ such that the Lie algebra of $G'$ is $\mathfrak{g}'$ and whose fundamental vector fields are given by $V$. 

Now, let $G$ be the Lie group associated with the Lie algebra $\mathfrak{g}$. The kernel $\ker\phi$ is an ideal of $\mathfrak{W}\simeq \mathfrak{g}$ that can be integrated to give a normal Lie subgroup $H$ of $G$. Hence, $G/H$ is a Lie group isomorphic to $G'$. Moreover, consider the Lie group projection $\pi:G\rightarrow G/H\simeq G'$, which is a Lie group morphism. This gives rise to a new Lie group action $\Phi:G\times \mm\rightarrow \mm$  such that $\Phi_g=\varphi_{\pi(g)}$ for every $g\in G$. It follows that the Lie algebra of fundamental vector fields of $\Phi$ is also $V$. 

Let us construct $J$ and prove that is a momentum map for $\Phi$. 
By assumption, $ h_1,\ldots,h_s $   are Casimir functions and  $ h_{s+1},\ldots,h_{s+r}$ complete them to give a basis of $\mathfrak{W}$ so that $s+r=p$. Then, the isomorphism between $\mathfrak{W}$ and $\mathfrak{g}$ shows that $ e_1,\ldots,e_s$ are related to functions corresponding to a zero vector field. Hence, $\langle e_1,\ldots,e_s\rangle$ span an ideal of $\mathfrak{g}$ contained in the centre of $\mathfrak{g}$. 
Define 
$$
J(x)=h_1(x)e^1+\ldots+h_s(x)e^s+h_{s+1}(x)e^{s+1}+\ldots+h_{s+r}(x)e^{s+r}
$$
 for the dual basis $ \{ e^1,\ldots,e^{s+r} \}$ to $\{e_1,\ldots,e_{s+r}\}$. Note that $\varphi$, and therefore $\Phi$, can be constructed so that   $\langle e_{s+1},\ldots,e_{s+r}\rangle$ have fundamental vector fields $\langle X_1,\ldots,X_r \rangle$. In fact, the commutation relations of $\{ X_1,\ldots,X_r\}$ and $\{ e_{s+1},\ldots,e_{s+r}\}$ are opposite, where $\{ e_{s+1},\ldots,e_{s+r}\}$ is the basis of $\mathfrak{g}'$ obtained by projecting $e_{s+1},\ldots,e_{s+r}$ in $\mathfrak{g}$ to $\mathfrak{g}'$. Recall that $\langle e_1,\ldots,e_s\rangle $ have zero fundamental vector fields. The coordinates of $J$ span a Lie algebra isomorphic to $\mathfrak{g}$ and $\langle J(x),e_\alpha\rangle=h_\alpha(x)$ for $\alpha=1,\ldots,s+r=p$ and every $x\in \mm$. Note that
$$
J_{[e_\alpha,e_\beta]}=\{J_{e_\alpha},J_{e_\beta}\},\qquad 1\leq \alpha<\beta\leq s+r.
$$
Hence, $J:G\rightarrow \mathfrak{g}^*$ is a momentum map for $\Phi$ and the result follows. 
\end{proof}

Proposition \ref{Prop:EqLieMom} has a relevant meaning: a finite-dimensional Lie algebra of Hamiltonian functions relative to a Poisson bracket always gives rise to an Ad-invariant momentum map, at least locally.

Let us comment on the inequality $  r\le p$  in Proposition \ref{Prop:EqLieMom}. When $\mathfrak{W}\simeq \mathfrak{g}$ is a  semisimple Lie algebra, then $s=0$ and $r=p$, as it has no ideals. In contrast, when $\mathfrak{W}\simeq \mathfrak{g}$ is not semisimple, it may admit one or more non-trivial central extensions, implying that it may happen that $s\ge 1$ and $ r<p$, as is the case for  $\mathfrak{sl}(2,\mathbb{R})\oplus \mathfrak{h}_3$  in the previous section (with $s=1$).
In this respect, let us  illustrate with more detail  the latter situation by focusing on $ \mathfrak{h}_3$.

%%%%%%%%%%%%%%%%%%%%%%%%%%%%%%

\subsection{Heisenberg Lie algebra}
\label{s41}

  Consider the functions $  h_0=1,h_1=q, h_2=p $ on ${\rm T}^*\mathbb{R}$ spanning a Lie algebra $\mathfrak{W}$  relative to the standard Poisson bracket in $C^\infty({\rm T}^*\mathbb{R})$ (as in (\ref{hamHH})). Then, $\mathfrak{W}$ is isomorphic to  the Heisenberg Lie algebra $\mathfrak{h}_3$  and the basis $\{ h_0,h_1,h_2\}$ is associated with the Hamiltonian vector fields on ${\rm T}^*\mathbb{R}$ given by $X_0=0, X_1=-\partial_p, X_2=\partial_q$, with respect to the standard symplectic structure on ${\rm T}^*\mathbb{R}$. The Lie algebra of vector fields $V=\langle X_0,X_1,X_2\rangle$ is isomorphic to the two-dimensional abelian Lie algebra, giving rise to a Lie group action of $\mathbb{R}^2$ on ${\rm T}^*\mathbb{R}$
$$
\varphi(\mu_1,\mu_2;q,p)=(q+\mu_2,p-\mu_1),\qquad \forall (\mu_1,\mu_2)\in \mathbb{R}^2, \qquad \forall (q,p)\in {\rm T}^*\mathbb{R},
$$
and whose fundamental vector fields are given by $V$. Note that the isomorphism $\phi:h\in \mathfrak{W}\mapsto -X_h\in V$ has $\ker \phi=\langle h_0\rangle$. The isomorphism   $\mathfrak{W}\simeq  \mathfrak{h}_3$ mapping $\phi(h_\alpha)=e_\alpha$ for a basis $\{ e_0,e_1,e_2\}$ of $\mathfrak{h}_3$ shows that $\langle e_0\rangle$ is an ideal of $\mathfrak{h}_3$. Note that one can assume a basis of $\mathfrak{h}_3$ as $3\times 3$ strictly upper triangular real matrices of the form 
$$
e_0=\left(\begin{array}{ccc}0&0&1\\0&0&0\\0&0&0\end{array}\right),\qquad
e_1=\left(\begin{array}{ccc}0&1&0\\0&0&0\\0&0&0\end{array}\right),\qquad e_2=\left(\begin{array}{ccc}0&0&0\\0&0&1\\0&0&0\end{array}\right),
$$
satisfying the  commutation relations $[e_1,e_2]=e_0$ and $[e_0,\cdot]=0$. Hence, consider the Lie group   ${\rm H}_3$  of $\mathfrak{h}_3$ given by the matrices
$$
{\rm H}_3=\left\{\left(\begin{array}{ccc}1&a&c\\0&1&b\\0&0&1\end{array}\right):a,b,c\in \mathbb{R}\right\}
$$
and the Lie subgroup related to the elements in the kernel of $\mathfrak{h}_3$  
$$
H=\left\{\left(\begin{array}{ccc}1&0&c\\0&1&0\\0&0&1\end{array}\right):c\in \mathbb{R}\right\}.
$$
Then, there exists a projection from the Heisenberg group to $\mathbb{R}^2$:
$$
\pi:{\rm H}_3=\left\{\left(\begin{array}{ccc}1&a&c\\0&1&b\\0&0&1\end{array}\right):a,b,c\in \mathbb{R}\right\}\mapsto (a,b)\in {\rm H}_3/H\simeq \mathbb{R}^2,
$$
which is also a surjective Lie group morphism.
We obtain a Lie group action 
$
\Phi:{\rm H}_3\times \mathbb{R}^2\rightarrow \mathbb{R}^2
$
such that
$$
\Phi\left(\left(\begin{array}{ccc}1&a&c\\0&1&b\\0&0&1\end{array}\right);(q,p)\right)=(q+b,p-a),\qquad \forall a,b,c\in \mathbb{R}, \qquad \forall q,p\in \mathbb{R},
$$
which has a  momentum map  given by
$$
J:(q,p)\in {\rm T}^*\mathbb{R}\mapsto e^0+qe^1+pe^2\in \mathfrak{h}_3^*,
$$
with $ \langle e^0,e^1,e^2 \rangle$ being the dual basis of $\langle e_0,e_1,e_2\rangle$. Indeed, 
$$
e^0+(q+c)e^1+(p-a)e^2=J\circ \Phi_A={\rm Ad}_{A^{-1}}^T\circ J,\qquad \forall A=\left(\begin{array}{ccc}1&a&c\\0&1&b\\0&0&1\end{array}\right)\in {\rm H}_3,
$$
 which shows that $\Phi$ is ${\rm Ad}$-invariant.

%%%%%%%%%%%%%%%%%%%%%%%%%%%%%%%%%%%%%%%%%%%

\section{Fundamental properties}
\label{s5}

We now focus our attention on the relations between nonlinear LH systems and other previous developments. We also illustrate the generality of the notion introduced in Definition~\ref{Def:nonLH}, which is potentially of use to study generic $t$-dependent Hamiltonian systems around a generic point.

Let $F\in C^\infty(\mathbb{R}\times \mathfrak{g}^*)$  be a $t$-dependent function on $\mathfrak{g}^*$ related to a nonlinear LH system $(\mm,\Lambda,X,J)$. The function $F$ leads to a $t$-dependent function $h$ on $\mm$ of the form $h(t,x)=F(t,J(x))$ for every $t\in \mathbb{R}$ and $x\in \mm$. More specifically, considering the coordinates of $J$ relative to a basis $ \{e^1,\ldots,e^r \}$  of $\mathfrak{g}^*$, one can write
\begin{equation}
\label{Eq:NonHamLie}
h(t,x)=F(t,J(x))=F\bigl(t,h_1(x),\ldots,h_r(x)\bigr),\qquad \forall t\in \mathbb{R},\qquad \forall x\in \mm,
\end{equation}
where $\{ h_1,\ldots,h_r\} $ form a basis of a Lie algebra of Hamiltonian functions relative to the Poisson bracket on $\mm$ given by $h_\alpha(x)=\langle J(x),e_\alpha\rangle$ for $\alpha=1,\ldots,r$ and any $x\in \mm$.

It is worth remarking  that (\ref{Eq:NonHamLie})  retrieves, when the dependence on $t$ is removed, the Hamiltonian functions studied in the theory of classical  integrable systems with Poisson coalgebra symmetries, as developed in many works (see \cite{LS20,BBR06,BH07,BBHMR09,BBH09,BB10} and references therein). This theory considers an $r$-dimensional Lie algebra $\langle h_1,\ldots,h_r\rangle$ of Hamiltonian functions on a Poisson manifold relative to its Poisson bracket, and studies the properties of Hamiltonian systems related to a function of the form $h=F(h_1,\ldots,h_r)$.

On the other hand, (\ref{Eq:NonHamLie}) is related to other theories. For instance,
$
h(t,x)=\sum_{\alpha=1}^rb_\alpha(t)h_\alpha,
$
with arbitrary $t$-dependent functions $\{b_1(t),\ldots,b_r(t)\}$, i.e.,~a linear combination with $t$-dependent coefficients of the functions $ h_1,\ldots,h_r $, gives rise to a Lie--Hamiltonian structure $(\mm,\Lambda,h)$, which is associated with LH systems, as presented in Definitions~\ref{defLH} and \ref{defLHS}.  Hence, this is a particular case of (\ref{Eq:NonHamLie}), which furthermore allows for considering much more general $t$-dependent functions of the basis $\{ h_1,\ldots,h_r\}$.  

Finally, the function $h$ in (\ref{Eq:NonHamLie}) can be considered as a $t$-dependent generalization of the so-called {\em collective Hamiltonians} \cite{GS80}. A collective Hamiltonian is a Hamiltonian function on a Poisson manifold $(\mm,\Lambda)$ of the form $h\circ J$, where $J$ is a certain momentum map $J:\mm\rightarrow \mathfrak{g}^*$. The most important difference of the definition of collective Hamiltonians with Definition \ref{Def:nonLH} is just the fact that the coordinates of $J$ do not need to span a Lie algebra of functions of dimension $\dim\mathfrak{g}$, as in our theory, but a Lie algebra of smaller dimension. Nevertheless, this is, in general, a technical issue and, in practice, one can say that both definitions are concerned with the same type of systems.

It should be noted  that this last connection does not seem to have been considered in the Poisson coalgebra symmetry theory, despite the relations to collective Hamiltonians that are to be displayed next. This important relation has also been omitted so far in the literature concerning the theory of LH systems (cf.~\cite{LS20}).

In view of Proposition \ref{Prop:EqLieMom},
the Lie algebra $\mathfrak{W}$ gives rise to a Lie algebra of Hamiltonian vector fields whose Hamiltonian functions are given by the elements of $\mathfrak{W}$. Locally, such vector fields can be integrated via Proposition \ref{Prop:EqLieMom} to a Hamiltonian Lie group action $\Phi:G\times \mm\rightarrow \mm$ with a momentum map $J:x\in \mm\mapsto \sum_{\alpha=1}^rh_\alpha(x)e^\alpha\in  \mathfrak{g}^*$ and $\mathfrak{g}\simeq \mathfrak{W}$.  Then, every $t$-dependent function $F\in C^\infty(\mathbb{R}\times \mathfrak{g}^*)$  yields a $t$-dependent Hamiltonian on $\mm$  of the form (\ref{Eq:NonHamLie}).
 In particular, standard collective Hamiltonians are recovered when $F$ is $t$-independent.

Let us now turn to study   general geometric properties of nonlinear LH systems.

\begin{proposition} 
\label{prop51}
Given a nonlinear LH system $(\mm,\Lambda,X,J:\mm\rightarrow \mathfrak{g}^*)$, the $t$-dependent vector field $X$ gives rise to a $t$-parametric family of vector fields on $\mm$ taking values in the generalized integrable distribution
$$
\mathcal{D}_x=\langle X_{h_1}(x),\ldots,X_{h_r}(x)\rangle,\qquad \forall x\in \mm,
$$
where $h_\alpha(x)=\langle J(x),e_\alpha\rangle$ for $\alpha=1,\ldots,r$, while $\{ e_1,\ldots,e_r\}$ is a basis for $\mathfrak{g}$.
\end{proposition}
\begin{proof}
    By assumption, $X$ has a $t$-dependent Hamiltonian function $h=F(t,h_1,\ldots,h_r)$ of the form (\ref{Eq:NonHamLie}). Hence, the relation $X_t=-\widehat{\Lambda}(\dd h_t)$ for every $t\in \mathbb{R}$ implies that 
\be
X(t,x)= \sum_{\alpha=1}^r\frac{\partial }{\partial e_\alpha}F\bigl(t,h_1(x),\ldots,h_r(x) \bigr)X_{h_\alpha}(x),\qquad \forall t\in \mathbb{R},\qquad \forall x\in \mm,
\label{5xx}
\ee
where $\{e_1,\ldots,e_r\}$ are considered as linear coordinates in $\mathfrak{g}^*$.
Consequently, each vector field $X_t$ takes values in the generalized distribution spanned by the vector fields $ X_{h_1},\ldots,X_{h_r} ,
$ namely
$$
\mathcal{D}_x=\langle X_{h_1}(x),\ldots,X_{h_r}(x)\rangle,\qquad \forall x\in \mm.
$$
Moreover, $\langle X_{h_1},\ldots,X_{h_r} \rangle$ is the Lie algebra of fundamental vector fields of the Lie group action associated with $J$. Hence, they span a generalized integrable distribution.
 \end{proof}

 Observe that for a standard (linear) LH system $h=\sum_{\alpha=1}^rb_\alpha(t)h_\alpha$, the relation~(\ref{5xx}) leads to the usual $t$-dependent vector field $X = \sum_{\alpha=1}^rb_\alpha(t)X_\alpha$, 
  with $  X_\alpha \equiv X_{h_\alpha}$, closing on the corresponding Vessiot--Guldberg Lie algebra of Hamilton vector fields.

Proposition \ref{Prop:GenNLH} below shows, on the one hand, how general nonlinear LH systems are. On the other hand, it proves that it is the mapping $J$, and its associated Lie algebra of functions, what really matters in practice to study $X$. For this purpose, we consider the  {\em Heisenberg Lie algebra in $n$ dimensions}, i.e., the $(2n+1)$-dimensional  Lie algebra $\mathfrak{h}(n)$ with a basis $\{ e_0,e_1,\ldots,e_{2n}\}$  and non-zero commutation relations 
$$
[e_{2i-1},e_{2i}]=e_0,\qquad i=1,\ldots,n.
$$
Here, $\mathfrak{h}(1) \equiv \mathfrak{h}_3$ corresponds to the usual Heisenberg Lie algebra.

 \begin{proposition}\label{Prop:GenNLH}For every $t$-dependent Hamiltonian vector field $X$ on a Poisson manifold $(\mm,\Lambda)$, there exists, around every generic point $x
\in \mm$ where $\Lambda$ is regular, an open neighbourhood $U$ of $x$ and an associated nonlinear LH system of the form $\bigl(U,\Lambda|_U,X|_U,J:U\rightarrow \mathfrak{h}(n)^*\oplus \mathbb{R}^{s*} \bigr)$, where $\dim \mm=2n+s$. 
 \end{proposition}
 \begin{proof}
  The Poisson bivector  $\Lambda$ gives rise to a local decomposition of $\mm$ into symplectic submanifolds. If the decomposition is a foliation, which happens on a neighbourhood $U$ of any point in $\mm$ where $\Lambda$ is regular (cf.~\cite{AM78}), then one can choose coordinates $\{q_1,\ldots,q_n,p_1,\ldots,p_n,z_1,\ldots,z_s\}$, where $2n+s=\dim \mm$, so that $\Lambda $ is given by (\ref{bivector}). Hence,  the Hamiltonian function of $X$ can always be written as $h=h(t,q_1,\ldots,q_n,p_1,\ldots,p_n,z_1,\ldots,z_s)$. But   $\langle q_1,\ldots,q_n,p_1,\ldots,p_n\rangle$ is a Lie algebra isomorphic to the $(2n+1)$-dimensional Heisenberg Lie algebra $\mathfrak{h}(n)$. In fact, the only non-zero commutation relations read
$\{q_i,p_j\}_\Lambda=\delta_{ij}$ for $i,j=1,\ldots,n$.
In addition, the functions $\langle z_1,\ldots,z_s\rangle$ Poisson commute among themselves and with all the other functions. Thus, $X$ can be considered as a nonlinear LH system and $\mathfrak{g}\simeq \mathfrak{h}(n)\oplus \mathbb{R}^s$. 
\end{proof}

Although every $t$-dependent vector field can be considered, under very mild conditions, as a nonlinear LH system, the way of doing it can be difficult as it relies on putting a Poisson bivector into canonical form. Moreover, the determination of  the coordinates $\{q_1,\ldots,q_n,p_1,\ldots,p_n,z_1,\ldots,z_s\}$ can be quite complicated, which makes the previous result insufficient to study certain $t$-dependent Hamiltonian systems. Finally, constants of the motion and other characteristics of nonlinear LH systems can be obtained provided the Lie algebra $\mathfrak{g}$ is of a particular type, e.g., simple or reductive. This also may restrict the range and applicability of Proposition \ref{Prop:GenNLH}.  
  
 %%%%%%%%%%%%%%%%%%%%%%%%
 
 \subsection{{\em t}-Dependent Hamiltonian formulation of Painlev\'e equations}
\label{s51}

As an application of the above results, let us consider  the case of the $t$-dependent Hamiltonian formulation for Painlev\'e invariants \cite{Pain,Decio}. The second Painlev\'e invariant is determined by the second-order differential equation in canonical form 
$$
\frac{\dd^2y}{\dd t^2}=2y^3+t y+b-\frac 12,\qquad\forall t\in \mathbb{R},\qquad \forall y\in \mathbb{R},\qquad b\in \mathbb{R} .
$$
This equation can be reformulated in terms of $t$-dependent Hamiltonian systems on ${\rm T}^*\mathbb{R}$ by considering 
$$
q=y, \qquad p=\frac{\dd y}{\dd t}+y^2+\frac{t}{2},
$$
 from which the equations of the motion 
$$
\left\{
\begin{aligned} 
\frac{\dd q}{\dd t}&=\frac{\partial h_{\rm II}}{\partial p}=p-q^2-\frac t2,\\
\frac{\dd p}{\dd t}&=-\frac{\partial h_{\rm II}}{\partial q}=2qp+b ,
\end{aligned}\right.
$$
are obtained with respect to the canonical symplectic form on ${\rm T}^*\mathbb{R}$ and the $t$-dependent Hamiltonian function given by
$$
h_{\rm II}(t,q,p)=\frac 12\, p\bigl(p-2q^2-t \bigr) -bq.
$$
Hence, we can consider the standard Lie group action of ${\rm H}_3$ on ${\rm T}^*\mathbb{R}$ of the form 
 $$
 \Phi:\bigl((\mu_0,\mu_1,\mu_2);q,p  \bigr)\in {\rm H}_3\times {\rm T}^*\mathbb{R} \mapsto (q+\mu_2,p-\mu_1) \in {\rm T}^*\mathbb{R}
 $$ 
 and then $J:(q,p)\in {\rm T}^*\mathbb{R}\mapsto  e^0+qe^1+p e^2\in  \mathfrak{h}(1)^* \equiv \mathfrak{h}_3^*$, as in Section~\ref{s41}. Therefore, $\mathfrak{W}\simeq \mathfrak{h}_3$ is spanned by 
 $\langle h_0=1,h_1=q, h_2=p\rangle$ such that  $\{h_1,h_2\}_\Lambda=h_0$ with respect to the 
 canonical symplectic form on ${\rm T}^*\mathbb{R}$. The associated Hamiltonian vector fields, via the canonical symplectic form on ${\rm T}^*\mathbb{R}$,  leads to  $V=\langle X_0=0, X_1=-\partial_p, X_2=\partial_q\rangle \simeq \mathbb{R}^2$.
 
 Finally, considering the function   $F\in C^\infty( \mathbb{R}\times \mathfrak{h}^*_3)$  given by
 $$
  F\bigl(t,\lambda_0e^0+\lambda_1e^1+\lambda_2e^2\bigr)=\frac 12 \lambda_2\bigl(\lambda_2-2\lambda_1^2-t\bigr) -b\lambda_1,
 $$
the $t$-dependent Painlev\'e  Hamiltonian is recovered  from the generic expression (\ref{Eq:NonHamLie})
 $$
F(t,J(q,p))= F(t, h_0,h_1,h_2)= \frac 12 h_2\bigl(h_2-2h_1^2-t\bigr) -bh_1 \equiv h_{\rm II}(t,q,p) ,
$$
allowing us to determine the corresponding $t$-dependent vector field $X_{\rm II}$  from the expression (\ref{5xx}), namely
\be
\begin{split}
X_{\rm II}& = - (2 h_1 h_2 + b)X_1+ \bigl(h_2 - h_1^2 -\tfrac 12 t \bigr) X_2\\[2pt]
&=   (2 q p+ b)\frac{\partial }{\partial p}+ \bigl(p - q^2 -\tfrac 12 t \bigr) \frac{\partial }{\partial q} .
\end{split}
\nonumber
\ee
 In this case, each vector field $(X_{\rm II})_t$  takes values in the generalized distribution   spanned by 
$X_1 ,  X_2$, according to Proposition~\ref{prop51}.

%%%%%%%%%%%%%%%%%%%%%%%%%%%%%%%%%%%%%%%%%%%

\section{Constants of the motion}
\label{s6}

  In this section, we  determine constants of the motion for nonlinear LH systems and analyze their properties.
In particular, we state that the space of constants of the motion for the $t$-dependent vector field $X$ of a nonlinear LH system $(\mm,\Lambda,X,J)$ admits a natural Poisson algebra structure, which may be used to derive new constants of the motion for $X$ from known ones. Among other achievements, our theory extends to nonlinear LH systems previous results relative to $t$-dependent and $t$-independent constants of the motion for LH systems, as described in \cite{BCHLS13,LS20} (the proofs of some of the following results can be found in these references). Moreover, we relate nonlinear LH systems with Hamiltonian systems on duals to Lie algebras.
 
A $t$-dependent constant of the motion for a system $X$ on a manifold $\mm$ is a first integral $\I\in C^\infty(\mathbb{R}\times \mm)$ of the autonomization $\bar X$ of $X$, namely
\begin{equation}
\bar X\I:=\frac{\partial \I}{\partial t}+X\I=0,
\label{sx}
\end{equation}
where $\bar X$ and $X$ are understood as vector fields on $\mathbb{R}\times \mm$ in the natural manner, while $t$ is the variable in $\mathbb{R}$. This fact leads to the following proposition, which was proved in~\cite[Proposition 8]{BCHLS13} and that  applies, in particular, to nonlinear LH systems.

\begin{proposition}\label{Prop:Ralg} The space $\bar{\mathcal{I}}^X$ of $t$-dependent constants of the motion for a system $X$ forms an $\mathbb{R}$-algebra $\bigl(\bar{\mathcal{I}}^X,\boldsymbol{\cdot}\, \bigr)$ relative to the product of functions ``$\,\boldsymbol{\cdot}$".
\end{proposition}

To generalize the second statement of Proposition \ref{IntLieHam0}   to $t$-dependent
 constants of the motion, the space $\mathbb{R}\times \mm$ can be endowed with a Poisson structure, which turns $\bar{\mathcal{I}}^X$ into a Poisson algebra (see~\cite[Lemma 9]{BCHLS13}).

\begin{lemma}\label{Le::NewBrak} Every Poisson manifold $(\mm,\Lambda)$ induces a Poisson manifold $\bigl(\mathbb{R}\times \mm,\bar\Lambda \bigr)$ with Poisson bracket
\begin{equation}
\label{newBrack}\nonumber
\{f,g\}_{\bar\Lambda}(t,x)=\{f_t,g_t\}_{\Lambda}(x),\qquad \forall (t,x)\in
\mathbb{R}\times \mm,\qquad \forall f,g\in C^\infty(\mathbb{R}\times \mm),
\end{equation}
where we define as previously $f_t:x\in M\mapsto f(t,x)\in \mathbb{R}$ for every $f\in C^\infty(\mathbb{R}\times M)$. 
\end{lemma}

This result motivates the following notion:

\begin{definition} Given a Poisson manifold $(\mm,\Lambda)$,
the associated Poisson manifold $\bigl(\mathbb{R}\times \mm,\bar \Lambda\bigr)$ is called the {\em autonomization} of $(\mm,\Lambda)$. Likewise,
 the Poisson bivector $\bar {\Lambda}$ is called the {\em autonomization} of $\Lambda$.
\end{definition}

The following lemma allows us to show that $\bigl(\bar{\mathcal{I}}^X,\boldsymbol{\cdot}\, ,\{\cdot,\cdot\}_{\bar\Lambda} \bigr)$
is a Poisson algebra~\cite{LS20}.

\begin{lemma}\label{AutLieHam} Let $(\mm,\Lambda)$ be a Poisson manifold and
let $X$ be a Hamiltonian vector field on $\mm$ relative to it. Then, the Lie derivative satisfies 
$\mathcal {L}_{\bar{X}}\bar{\Lambda}=0$.
\end{lemma}

This enables us to formulate and prove a slight generalization to nonlinear LH systems from a similar result that holds for LH systems~\cite[Proposition 12]{LS20}.

\begin{proposition}\label{AutLieHam2}  Let $(\mm,\Lambda,X,J)$ be a nonlinear LH system. Then, $\bigl(\overline{\mathcal{I}}^X,\boldsymbol{\cdot}\, ,\{\cdot,\cdot\}_{\bar\Lambda} \bigr)$
is a Poisson algebra.
\end{proposition}

\begin{proof} Proposition \ref{Prop:Ralg} ensures that $\bigl(\bar{\mathcal{I}}^X,\boldsymbol{\cdot}\, \bigr)$ is an $\mathbb{R}$-algebra. To prove that $\bigl(\bar{\mathcal{I}}^X,\boldsymbol{\cdot}\, ,\{\cdot,\cdot\}_{\bar\Lambda} \bigr)$ is a Poisson algebra, let us show that  $\bar{\mathcal{I}}^X$ is closed with respect to the Poisson bracket $\{\cdot,\cdot\}_{\bar{\Lambda}}$, i.e.~$\bar{X}\{f,g\}_{\bar{\Lambda}}=0$ for all $f,g\in \bar{\mathcal{I}}^X$.
 As the vector fields $\{X_t\}_{t\in\mathbb{R}}$ are Hamiltonian relative to $(\mm,\Lambda)$, and using 
 Lemma \ref{AutLieHam}, we have that $\bar\Lambda$ is invariant under the autonomization of each vector field $X_{t'}$ with $t'\in\mathbb{R}$, i.e.~$
 \mathcal{L}_{ \bar{X}_{t'}}\bar\Lambda=0$. Therefore,
 \begin{equation}
\begin{split}
   \bar X\{f,g\}_{\bar{\Lambda}}(t',x)&=\bar {X}_{t'}\{f,g\}_{\bar{\Lambda}}(t',x)= \bigl\{\bar{X}_{t'}f,g \bigr\}_{\bar{\Lambda}}(t',x)+
\bigl \{f,\bar{ X}_{t'}g \bigr\}_{\bar{\Lambda}}(t',x)\\[2pt]
 &=\{\bar{X}f,g\}_{\bar{\Lambda}}(t',x)+
 \{f,\bar{X}g\}_{\bar{\Lambda}}(t',x)=0.
\end{split}
\nonumber
\end{equation}
It follows at once from this relation that $\{f,g\}_{\bar{\Lambda}}$ is a $t$-dependent constant of the motion for $X$.
 \end{proof}

The following statement will be a key to obtain $t$-independent constants of the motion for nonlinear LH systems. In particular, it generalizes Proposition \ref{IntLieHam0} to the case of nonlinear LH systems and, furthermore,  also constitutes the cornerstone for extending  the Poisson coalgebra method to nonlinear LH systems, which will be addressed in the next section.

\begin{proposition} 
\label{prop63}
Let $(\mm,\Lambda,X,J)$ be a nonlinear LH system such that $X_t$ admits a Hamiltonian function $h_t$ for every $t\in \mathbb{R}$. A function $\I\in C^\infty(\mm)$ is a constant of the motion for $X$ if and only if $\I$ Poisson commutes with all elements of ${\rm Lie}\bigl(\{h_t\}_{t\in \mathbb{R}}, \{\cdot,\cdot\}_\Lambda \bigr)$.
\end{proposition}

\begin{proof} The function $\I$ is a $t$-independent constant of the motion for $X$ if and only if
\begin{equation}
0 = X_t \I = \{\I,h_t\}_{\Lambda}, \qquad \forall t \in \mathbb{R}. 
\nonumber
\end{equation}
Hence,
$$
\bigl\{\I,\{h_t,h_t'\}_\Lambda \bigr\}_\Lambda =\bigl\{\{\I,h_t\}_{\Lambda},h_t' \bigr\}_\Lambda +\bigl\{h_t,
\{\I,h_t'  \}_\Lambda \bigr\}_\Lambda =0, \qquad \forall t,t' \in \mathbb{R},
$$
and by recursion, it follows that $\I$ Poisson commutes with all successive Poisson brackets of elements of $\{h_t\}_{t\in \mathbb{R}}$ and their linear combinations. As these elements span ${\rm Lie}\bigl(\{h_t\}_{t\in \mathbb{R}}, \{\cdot,\cdot\}_\Lambda \bigr)$, we conclude that $\I$ Poisson commutes with Lie$\bigl(\{h_t\}_{t\in \mathbb{R}} , \{\cdot,\cdot\}_\Lambda\bigr)$.
Conversely, if $\I$ Poisson commutes with Lie$\bigl(\{h_t\}_{t\in \mathbb{R}}, \{\cdot,\cdot\}_\Lambda \bigr)$, it Poisson commutes with the elements $\{h_t\}_{t\in \mathbb{R}}$, and, in view of \eqref{sx}, it becomes a constant of the motion for $X$.
\end{proof}

Nonlinear LH systems are naturally related to the duals to Lie algebras by means of a momentum map, although we have not described so far the precise role that such duals play in the study of nonlinear LH systems.  In the following, we prove that any nonlinear LH system is associated with a $t$-dependent Hamiltonian system on the dual to a Lie algebra with respect to the KKS  bracket. 

Recall that each nonlinear LH system $(\mm,\Lambda,X,J:\mm\rightarrow \mathfrak{g}^*)$ is related to a $t$-dependent function  $F\in C^\infty(\mathbb{R}\times\mathfrak{g}^*)$ given in (\ref{Eq:NonHamLie}). Then, one can define, for every value $t\in \mathbb{R}$, a function  $F_t:  \mathfrak{g}^* \rightarrow \mathbb{R}$. The functions $\{F_t\}_{t\in \mathbb{R}}$ give rise to a $t$-dependent Hamiltonian related to a $t$-dependent vector field $X^{\mathfrak{g}^*}$ on $\mathfrak{g}^*$ relative to the KKS bracket on $\mathfrak{g}^*$ (see (\ref{biv})). Thus, $F$ yields a $t$-dependent Hamiltonian system on $\mathfrak{g}^*$ from the $t$-dependent vector field $X^{\mathfrak{g}^*}$ induced by   $\bigl\{X^{\mathfrak{g}^*}_t \bigr\}_{t\in \mathbb{R}}$. This leads to the following result.

\begin{proposition} 
\label{prop64}
A nonlinear LH system $(\mm,\Lambda,X,J:\mm
\rightarrow \mathfrak{g}^*)$ associated with a function $F\in C^\infty(\mathbb{R}\times\mathfrak{g}^*)$ yields a $t$-dependent Hamiltonian system $X^{\mathfrak{g}^*}$ on $\mathfrak{g}^*$ linked to $F$ via the KKS bracket on $\mathfrak{g}^*$, namely $X^{\mathfrak{g}^*}_t=-\widehat{
\Lambda}_{\mathfrak{g}^*}(\mbox{\rm \dd} F_t)$, for every $t\in \mathbb{R}$.
\end{proposition}

The properties of the system $X^{\mathfrak{g}^*}$ are fundamental to understanding those of the nonlinear LH system $(\mm,\Lambda,X,J:\mm\rightarrow \mathfrak{g}^*)$.  In particular,   if $S$ is a constant of the motion of  $X^{\mathfrak{g}^*}$, then $J^*S$ is a constant of the motion of $X$ as shown next.

\begin{proposition} 
\label{prop65}
Let $(\mm,\Lambda,X,J:\mm\rightarrow \mathfrak{g}^*)$ be a nonlinear LH system and let $X^{\mathfrak{g}^*}$ be its associated Hamiltonian system on $\mathfrak{g}^*$. Then, $S\in C^\infty(\mathbb{R}\times \mathfrak{g}^*)$ is a $t$-dependent constant of the motion of $X^{\mathfrak{g}^*}$ on $J(\mm)$ if and only if $J^*S$ is a $t$-dependent constant of the motion of $X$.
\end{proposition}

\begin{proof} As $J_*\Lambda_x=\Lambda_{\mathfrak{g}^*,J(x)}$  holds for every $x\in \mm$, we have  that $J_*{X_t}=X_t^{\mathfrak{g}^*}$ for every $t\in \mathbb{R}$. Hence, 
$$
\left(\frac{\partial S}{\partial t}+X^{\mathfrak{g}^*}S\right)(t,J(x))=\frac{\partial S_t}{\partial t}(J(x))+(X^{\mathfrak{g}^*}_tS_t)(J(x))=\frac{\partial J^*S_t}{\partial t}(x)+(X_tJ^*S_t)(x), 
$$
for all $(t,x)\in \mathbb{R}\times \mm$, from which the assertion follows.
\end{proof}

%%%%%%%%%%%%%%%%%%%%%%%%%%%%%%%%%%%%%%%%%%%

\subsection{Isotropic harmonic oscillator with {\em t}-dependent frequency}
\label{s61}

Let us illustrate Propositions~\ref{prop63}--\ref{prop65} through the $n$-dimensional isotropic harmonic oscillator (HO)  with a $t$-dependent frequency. We recall that the $t$-dependent HO is defined as the Hamiltonian system on ${\rm T}^*\mathbb{R}^n$ given by
\be
h_{\rm HO}(t,\>q,\>p)=   \frac 12 \sum_{i=1}^{n} p_i^2 +\frac 12 \omega^2(t) \sum_{i=1}^{n}q_i^2  ,
\label{hamHO}
\ee
where $\>q=(q_1,\dots,q_n)$ and $\>p=(p_1,\dots,p_n)$, which yields
 the following system of  differential equations
\be
\frac{\dd q_i}{\dd t}=p_i,\qquad \frac{\dd p_i}{\dd t}=-\omega^2(t)q_i,\qquad i=1,\dots, n.
\nonumber
\ee
Hence, it is associated with the $t$-dependent vector
field
\be
X_{\rm HO}(t,\>q,
\>p)=\sum_{i=1}^n\left(p_i\frac{\partial}{\partial
q_i} -\omega^2(t)q_i \frac{\partial}{\partial
p_i}\right)
\label{xswB}
\ee
with respect to the standard Poisson bivector  $\Lambda$ given in  (\ref{bivector}).

Observe that  the $t$-dependent HO is just the particular case of the  $t$-dependent SW system in Section~\ref{s21} with all $c_i=0$, hence without any Rosochatius--Winternitz potential. Thus, $h_{\rm HO}(t,\>q,\>p) = h_3+\omega^2(t)h_1$, 
where  $  h_1,h_2,h_3$, given by (\ref{hSW}) for all $ c_i=0$, span a Lie algebra of functions   isomorphic to $\mathfrak{sl}(2,\mathbb{R})$ satisfying  the Poisson brackets (\ref{commSW}) with respect to   $\Lambda$.   In these conditions, we can define a momentum map
$$
J:(\>q,\>p)\in {\rm T}^*\mathbb{R}^n\mapsto h_1(\>q,\>p)
e^1+h_2(\>q,\>p)e^2+h_3(\>q,\>p)e^3\in \mathfrak{sl}(2,\mathbb{R})^*,
$$
where $\{ e^1,e^2,e^3\}$ is the dual basis of   $\{ e_1,e_2,e_3\}$ such that the latter satisfy opposite commutation rules to (\ref{xSW}) (and formally similar to (\ref{commSW})), namely
\be
[e_1,e_2]=-e_1,\qquad  [e_1,e_3]=- 2 e_2,\qquad [e_2,e_3]=-e_3 .
\label{SWe}
\ee
If we consider  the $t$-dependent linear function
$$
F:\bigl(t,\lambda_1e^1+\lambda_2e^2+\lambda_3e^3 \bigr)\in \mathbb{R}\times \mathfrak{sl}(2,\mathbb{R})^*\mapsto \lambda_3+\omega^2(t)\lambda_1\in\mathbb{R},
$$
 it follows that the $t$-dependent HO Hamiltonian (\ref{hamHO}) is reproduced as
 $$
 h_{\rm HO}(t,\>q,\>p)=F(t,J(\>q,\>p))=h_3(\>q,\>p) + \omega^2(t)h_1(\>q,\>p) .
 $$
 Note also that  the associated  $t$-dependent vector field $X_{\rm HO}$ (\ref{xswB}) can be  recovered consistently from  the relation (\ref{5xx}), i.e.,
 $$
 X_{\rm HO}= X_3 +  \omega^2(t) X_1,
 $$
 where $X_1$ and  $X_3 $ are given in (\ref{SWxs}) for all $c_1=\ldots=c_n=0$.  Therefore, the $t$-dependent HO  can be  considered as   the  nonlinear LH system $\bigl({\rm T}^*\mathbb{R}^n,\Lambda,X_{\rm HO},J:{\rm T}^*\mathbb{R}^n\rightarrow \mathfrak{sl}(2,\mathbb{R})^* \bigr)$ with  associated $t$-dependent function $F$.

In addition, from Proposition~\ref{prop64} we find that $F=e_3+\omega^2(t) e_1$ induces the $t$-dependent Hamiltonian system on ${\rm T}^*\mathbb{R}^n$
\be
X^{\mathfrak{sl}(2,\mathbb{R})^{\ast}}=-2e_2\frac{\partial}{\partial e_1}-e_3\frac{\partial}{\partial e_2}+\omega^2(t)\left(e_1\frac{\partial}{\partial e_2}+2e_2\frac{\partial}{\partial e_3}\right)
\nonumber
\ee
relative to the Poisson bivector (\ref{biv})
\be
\Lambda_{\mathfrak{sl}(2,\mathbb{R})^{\ast}}=-e_1\frac{\partial}{\partial e_1}\wedge \frac{\partial}{\partial e_2}-2e_2\frac{\partial}{\partial e_1}\wedge \frac{\partial}{\partial e_3}-e_3\frac{\partial }{\partial e_2}\wedge \frac{\partial}{\partial e_3}.
\nonumber
\ee
A short calculation further shows that
$$
{\rm T}_{(\>q,\>p)}J(X_{\rm  HO})_t=(X^{\mathfrak{sl}(2,\mathbb{R})^{\ast}})_t({J(\>q,\>p)}),\qquad \forall (\>q,\>p)\in {\rm T}^*\mathbb{R}^n.
$$
The above expression amounts to the fact that if $(\>q(t),\>p(t))$ is a particular solution to $X_{\rm HO}$, then $J(\>q(t),\>p(t))$ is a particular solution to  $X^{\mathfrak{sl}(2,\mathbb{R})^*}$.

   The Casimir invariant of the Lie algebra $\mathfrak{sl}(2,\mathbb{R})$ with commutators (\ref{SWe}) provides a   constant of the motion of $X^{\mathfrak{sl}(2,\mathbb{R})^{\ast}}$, explicitly 
  $$
\C= e_1e_3 -e_2^2 ,
  $$  
 and from Proposition~\ref{prop65}, it gives rise to a $t$-independent constant of the motion of $X_{\rm HO}$ (\ref{xswB}) as $\I =J^*\C$:
 $$
 \I(\>q,\>p)=h_1h_3-h_2^2=\frac 14 \sum_{1\le i<j}^n \bigl(q_i p_j-q_j p_i \bigr)^2 ,
 $$
 recovering the well-known angular momentum symmetry of  the isotropic HO system. Note that $\I$ Poisson commutes with the Hamiltonian functions $\langle  h_1,h_2,h_3\rangle$, as expected by Proposition~\ref{prop63}.
 
This construction, in turn, can be seen as an example of the nonlinear LH formalism applied to (linear) LH systems. Notwithstanding, if we restrict now to the one-dimensional case, setting $n=1$ above, the $t$-dependent   HO can be studied via a proper nonlinear LH system in another manner. Consider the   Lie algebra   of functions on ${\rm T}^*\mathbb{R}$ spanned by 
$$ 
g_1=  q , \qquad g_2=p, \qquad  g_3= qp_ , \qquad  g_0=1 ,
$$
with commutation relations
\be
\{g_3,g_1\}_\Lambda= - g_1,\qquad \{g_3,g_2\}_\Lambda=  g_2,\qquad \{g_1,g_2\}_\Lambda= g_0,\qquad \{g_0,\cdot\,\}_\Lambda=0,
\label{h4com}
\ee
again with respect to the standard Poisson bracket  $\Lambda$ on ${\rm T}^*\mathbb{R}$. Then, $\langle g_0,g_1 ,g_2,g_3\rangle$ is  the oscillator Lie algebra $\mathfrak{h}_4 $, with $g_3$  playing the role of the ``counting" function, while $g_1$, $g_2$ behave as ``ladder" functions. 
Their associated Hamiltonian vector fields are obtained through the canonical symplectic form on ${\rm T}^*\mathbb{R}$, via (\ref{conditions}), giving   $V=\langle X_0=0, X_1=-\partial_p, X_2=\partial_q,X_3=q \partial_q- p \partial_p\rangle \simeq \mathfrak{iso}(1,1)$, i.e., the (1+1)-dimensional Poincar\'e algebra.
  
Thus $\mathfrak{h}_4 $ contains as a Lie subalgebra the Heisenberg Lie algebra $\mathfrak{h}_3$ considered in Section~\ref{s41} and applied to the second Painlev\'e invariant  in Section~\ref{s51}. As a consequence, we have an alternative momentum map 
 $$
 J':(q,p)\in{\rm T}^*\mathbb{R}\mapsto g_0(q,p)e^0+g_1(q,p)e^1+g_2(q,p)e^2+g_3(q,p)e^3\in \mathfrak{h}_4^* ,
 $$
 where  $\{ e^0,e^1,e^2,e^3\}$ is the dual basis to $\{ e_0,e_1,e_2,e_3\}$ in $\mathfrak{h}_4$, the latter satisfying  commutation rules formally similar to (\ref{h4com}).  Then, by introducing the function 
 $$
  F': \bigl(t, \lambda_0 e^0+\lambda_1 e^1+\lambda_2 e^2+\lambda_3 e^3 \bigr)\in \mathbb{R}\times \mathfrak{h}_4^*\mapsto \frac 12 \lambda_2^2 + \frac 12 \omega^2(t)\lambda_1^2 \in\mathbb{R}, 
  $$
  we can express the $t$-dependent HO Hamiltonian (\ref{hamHO}) for $n=1$ in a second form
\be
h_{\rm HO}(t,q,p)=F'(t,J'(q,p))=\frac 12 g_2^2(q,p) +\frac 12 \omega^2(t)g_1^2(q,p) =  \frac 12 p^2+\frac 12 \omega^2(t) q^2,
\label{h4HO}
\ee
and the corresponding  $t$-dependent vector field $X_{\rm HO}$ (\ref{xswB}) ($n=1$) can be  retrieved   applying (\ref{5xx}), 
 $$
 X_{\rm HO}=    \omega^2(t) g_1 X_1 + g_2 X_2 =- \omega^2(t) q \frac{\partial }{\partial p} + p \frac{\partial }{\partial q}.
 $$
We conclude that the HO  can also be considered as a nonlinear LH system $\bigr({\rm T}^*\mathbb{R},\Lambda,X_{\rm HO},J':{\rm T}^*\mathbb{R}\rightarrow \mathfrak{h}_4^*\bigl)$ with associated $t$-dependent function $F'$.  
 
Moreover, the system induced in $\mathfrak{h}_4^*$ by the function   $F' =\frac 12 e_2^2  +\frac 12 \omega^2(t)e_1^2$ and the   Poisson bivector (\ref{biv})
$$
\Lambda_{\mathfrak{h}_4^*}=e_0\frac{\partial}{\partial e_1}\wedge\frac{\partial}{\partial e_2}+e_1\frac{\partial}{\partial e_1}\wedge \frac{\partial}{\partial e_3}- e_2\frac{\partial}{\partial e_2}\wedge\frac{\partial}{\partial e_3}
$$
reads
$$
X^{\mathfrak{h}_4^*}=e_0e_2\frac{\partial}{\partial e_1}+e_2^2\frac{\partial}{\partial e_3}-\omega^2(t)\left( e_0e_1\frac{\partial}{\partial e_2}+e_1^2\frac{\partial}{\partial e_3}\right).
$$
Meanwhile, the induced system on ${\rm T}^*\mathbb{R}$ is  (\ref{h4HO}). Concerning the constants of the motion derived from this second construction, we find that the non-trivial Casimir invariant of the Lie algebra $\mathfrak{h}_4$ in the basis $\{ e_0,e_1,e_2,e_3\}$ (with formal commutation rules similar to (\ref{h4com})) is given by  
 $$
  \C'=e_0e_3-e_1 e_2 ,
  $$ 
 yielding a constant of the motion of  $X^{\mathfrak{h}_4^{\ast}}$. 
   
In conclusion, the one-dimensional HO with $t$-dependent frequency admits two different approaches as a nonlinear LH system, based in the Lie algebras  $\mathfrak{sl}(2,\mathbb{R})$  and  $\mathfrak{h}_4$.

%%%%%%%%%%%%%%%%%%%%%%%%%%%%%%%%%%%%%%%%%%%

\section{The Poisson coalgebra formalism  applied  to integrability}
\label{s7}

The Poisson coalgebra method was formerly introduced in~\cite{BBR06} as a procedure for the construction of constants of the motion for $t$-independent Hamiltonian systems relative to Poisson manifolds, and was also shown to be a powerful tool to obtain completely integrable systems in the Liouville sense~\cite{Perelomov}. It is worth recalling that this formalism can be applied to classical and quantum integrable systems, as well as to their quantum deformations~\cite{BBR06,Chains1999}.
  Later on, it was proven that  this approach provides, in fact,  {\em quasi-maximal superintegrability}~\cite{BHMR04}, meaning that for an $n$-dimensional classical Hamiltonian, not only  $(n- 1)$ functionally independent constants of the motion in involution can be found (besides the Hamiltonian itself), but  $(2n- 3)$ independent ones, so that there is only one remaining constant of the motion needed to reach {\em maximal superintegrability}. The latter does not come from the coalgebra approach but from ``hidden" symmetries such as the Laplace--Runge--Lenz vector. For applications of this formalism  to superintegrable systems we refer to~\cite{BBHMR09,BH07,BBH09,BB10} and references therein.
  
    In the context of LH systems, the Poisson coalgebra method was   adapted to obtain constants of the motion and superposition rules in~\cite{BCHLS13}, thus generalizing the formalism  to deal with Lie systems admitting a Vessiot--Guldberg Lie algebra of Hamilton vector fields relative to more general geometric structures.  In addition, the coalgebra superintegrability for LH systems was recently established in~\cite{BCFHL21}.
  In this section, we extend the Poisson coalgebra formalism to the very general realm of nonlinear LH systems. 
    
 Let $\{ v_1,\dots,v_r\} $ be a basis spanning a Lie--Poisson algebra  $\mathfrak{g}$ satisfying the commutation relations
   \be
\bigl\{v_\alpha,v_\beta \bigr\}_{\mathfrak{g}^{*}} = \sum_{{\gamma=1}}^{r} c_{\alpha\beta}^{\gamma}v_{\gamma},\qquad  \alpha,\beta=1,\dots, r,
\label{71}
\ee
   and let $S(\mathfrak{g})$ be its associated symmetric algebra, i.e.,~the space of commutative polynomials in the elements of the chosen basis. As in the previous sections, the isomorphism $\mathfrak{g}\simeq \mathfrak{g}^{**}$ allows us to consider $  v_1,\ldots,v_r $ as functions on $\mathfrak{g}^*$ and then $S(\mathfrak{g})$ can be understood as the subspace of polynomial functions of $C^\infty(\mathfrak{g}^*)$. Since $C^\infty(\mathfrak{g}^*)$ admits the KKS  bracket $\{\cdot,\cdot\}_{\mathfrak{g}^*}$, which for polynomial functions on $  v_1,\ldots,v_r $  is also polynomial, the Poisson bracket can be restricted to $S(\mathfrak{g})$ endowing it with a natural Poisson algebra structure with Poisson bracket $\{\cdot,\cdot\}_{
S(\mathfrak{g})}$. Moreover,  $S(\mathfrak{g})$ admits a  trivial (non-deformed) Poisson coalgebra structure $\bigl(S(\mathfrak{g}),\Delta\bigr)$ given by the coproduct  map   ${\Delta} : S(\mathfrak{g})\rightarrow
S(\mathfrak{g}) \otimes S(\mathfrak{g})$    defined by
 \begin{equation}
{\Delta}(v)=v\otimes 1+1\otimes v,\qquad \forall  v\in\mathfrak {g}\subset S(\mathfrak{g}),
 \label{72}
\end{equation}
which is a Poisson algebra homomorphism 
$$
 \Delta\bigl( \{u,v\}_{S(\mathfrak{g})} \bigr) = \bigl\{  \Delta(u),\Delta(v)
\bigr\}_{S(\mathfrak{g}) \otimes S(\mathfrak{g}) },\qquad \forall  u,v\in\mathfrak {g}\subset S(\mathfrak{g}),
$$
such that
$$
 \bigl\{ v_\alpha\otimes v_\beta,  v_\gamma\otimes v_\delta  \bigr\}_{S(\mathfrak{g}) \otimes S(\mathfrak{g}) }=
\bigl\{v_\alpha, v_\gamma \bigr\}_{S(\mathfrak{g})}\otimes v_\beta  v_\delta +
 v_\alpha  v_\gamma \otimes \bigl\{v_\beta, v_\delta \bigr\}_{S(\mathfrak{g})}  .
$$
As for any Hopf algebra,  the  coproduct $\Delta\equiv \Delta^{(2)}$ must fulfil  the coassociativity condition (see \cite{CP94,Abe} for details)
 \begin{equation}
({\rm Id} \otimes \Delta) \circ \Delta=(\Delta \otimes {\rm Id}) \circ \Delta ,
\label{73}
 \end{equation}
 which for the non-deformed case with (\ref{72})  yields  the third-order coproduct   
 $\Delta^{(3)}: S(\mathfrak{g})\rightarrow
S(\mathfrak{g}) \otimes S(\mathfrak{g})\otimes S(\mathfrak{g})\equiv   S^{(3)}(\mathfrak{g})$ uniquely defined by the condition
$$
{\Delta}^{(3)}(v)=v\otimes 1\otimes 1 +1\otimes v\otimes 1+1\otimes 1\otimes v,\qquad \forall  v\in\mathfrak {g}\subset S(\mathfrak{g}).
\label{3co}
$$
The coassociativity relation (\ref{73}) enables us to obtain recursively a $k^{\rm th}$-order coproduct map,   
$$
\Delta ^{(k)}: \ S(\mathfrak{g}) \rightarrow {\stackrel{ k\ {\rm times} } {\overbrace{   S(\mathfrak{g}) \otimes\cdots\otimes   S(\mathfrak{g})  }}  }\equiv  S^{(k)}(\mathfrak{g}),\qquad k\ge 2 ,
$$
in two ways:
\be
\begin{split}
& \bullet  \mbox{ Left-coproduct:} \qquad {\Delta}_{\rm L}^{(k)}:= \bigr(\,{\stackrel{(k-2)\ {\rm times}}{\overbrace{\, {\rm
Id}\otimes\cdots\otimes{\rm Id}\, }}}    \otimes {\Delta^{(2)}}  \bigr)\circ \Delta_{\rm L}^{(k-1)} .\\[2pt]
& \bullet  \mbox{ Right-coproduct:}   \quad\  {\Delta}_{\rm R}^{(k)}:= \bigr(  {\Delta^{(2)}} \otimes  {\stackrel{(k-2)\ {\rm times}}{\overbrace{\,{\rm
Id}\otimes\cdots\otimes{\rm Id} \, }}}     \,\bigr) \circ \Delta_{\rm R}^{(k-1)} .
 \end{split}
 \label{74}
\ee
For $k=2$, they are just the coproduct $\Delta\equiv \Delta^{(2)}$ in (\ref{72}).
 It is straightforward to verify that both $\Delta_{\rm L}^{(k)}$ and $\Delta_{\rm R}^{(k)}$ are also    Poisson algebra  morphisms. Moreover, both  $k^{\rm th}$-order coproducts  can be embedded into a bigger $(m+1)^{\rm th}$-order  tensor product space in the form~\cite{BHMR04,BBHMR09}
\be
\begin{split}
&  {\Delta}_{\rm L}^{ (k) }: \quad   S^{(k)}(\mathfrak{g}) \otimes\! \!
 {\stackrel{ {(m+1-k)} \ {\rm times} } {\overbrace{1\otimes \cdots\otimes 1}}} \quad \sim \quad {\rm space}\ 1\otimes  2\otimes 
\cdots \otimes k ,\\[2pt]
&    {\Delta}_{\rm R}^{(k)}: \quad   {\stackrel{ {(m+1-k)} \ {\rm times} } {\overbrace{1\otimes \cdots\otimes 1}}}  \! \!  \otimes S^{(k)}(\mathfrak{g})  \quad \sim \quad {\rm space}\ (m-k+2)\otimes(m-k+3)\otimes   \cdots\otimes (m+1), 
 \end{split}
\nonumber
\ee
  for $k=2,\dots,m+1$, such that  $ {\Delta}_{\rm L}^{ (m+1) }\equiv  {\Delta}_{\rm R}^{ (m+1) }$.
 
 Now,  let $X$ be a nonlinear LH system related to a LH algebra ${\cal{H}}_\Lambda$ spanned by the basis given by the linearly independent Hamiltonian functions  $ \{h_1,\dots,h_r\}$ verifying the same commutation relations (\ref{71}) with respect to a Poisson bivector $\Lambda$ on an $n$-dimensional  manifold $\mm$ with coordinates $x = \{ x^1 ,\dots,  x^n \}$.  Hence,  we define  the Lie algebra morphism  
\be
\psi:\mathfrak{g} \rightarrow {\cal{H}}_\Lambda\subset  C^\infty(\mmm),\qquad h_\alpha:=\psi(v_\alpha),\qquad \alpha=1,\ldots,r ,
\nonumber
\ee 
and  introduce a family of Poisson algebra morphisms  
  \be
\begin{split}
D:&\  P\in S(\mathfrak{g}) \mapsto J^*P\in C^\infty(\mmm) ,\\[2pt]
 D^{(k)}:&\  P\in S^{(k)}(\mathfrak{g}) \mapsto
{\stackrel{ k\ {\rm times} }  {\overbrace{    (J\otimes  {\cdots} \otimes J)^* }   }  } P\in  
 {\stackrel{ k\ {\rm times} }  {\overbrace{   C^\infty(\mmm) \otimes {\cdots}\otimes   C^\infty(\mmm)   }   }  }\subset  C^\infty \bigl(\mmm^{k} \bigr),
 \end{split}
 \label{76}
\ee
where   $k\ge 2$.   By taking into account that the coproduct is primitive, together with the expressions (\ref{74}) and (\ref{76}), we define the following Hamiltonian functions on $  C^\infty \bigl(\mmm^{m+1} \bigr),$
  \be
\begin{split}
h^{(1)}_\alpha&:=D( v_\alpha)= h_\alpha \bigl( x_{(1)} \bigr)   , \qquad   \alpha=1,\dots , r ,\\[2pt]
 h^{(k)}_{{\rm L},\alpha}&:=  
 D^{(k)} \!\left( \Delta_{\rm L}^{(k)}(v_\alpha) \right)=\sum_{\ell =1}^k h_\alpha\bigl(x_{(\ell)}\bigr) ,\\[2pt]
 h^{(k)}_{{\rm R},\alpha}&:=    D^{(k)} \!\left( \Delta_{\rm R}^{(k)}(v_\alpha) \right)=\sum_{\ell =m-k+2}^{m+1} h_\alpha\bigl(x_{(\ell)}\bigr) ,   \end{split}
 \label{77}
\ee
where  $x_{(\ell)}=\bigl\{ x^1_{(\ell)},\dots,  x^n_{(\ell)}\bigr\}$ denotes the coordinates in the  $\ell$-th copy  manifold $\mm$ within $ {\mmm}^{m+1}$ and 
\be
h^{(m+1)}_{{\rm L},\alpha}= h^{(m+1)}_{{\rm R},\alpha}\equiv 
 h^{(m+1)}_{\alpha},\qquad \alpha = 1,\dots, r.
\nonumber
\ee

By construction, for a given value of $k$, each set  
$h^{(k)}_{{\rm L},\alpha}$ and $h^{(k)}_{{\rm R},\alpha}$ close on  the commutation rules (\ref{71}) with respect to the Poisson bivector $\Lambda^{m+1}$ on $ {\mmm}^{m+1}$ provided by $\Lambda$ in ${\cal{H}}_\Lambda$, so that  a Poisson bracket induced by $\Lambda^{m+1}$  is given by
\be
 \bigl\{\cdot\,,\cdot \bigr\}_{\Lambda^{m+1}}:C^\infty \bigl(\mmm^{m+1} \bigr)\times C^\infty \bigl(\mmm^{m+1} \bigr)\rightarrow C^\infty \bigl(\mmm^{m+1} \bigr).
\label{78}
\ee

We already have all the ingredients to establish the main result of this section.

\begin{theorem} 
\label{teor7}
Let $(\mm,\Lambda,X,J:\mm\rightarrow \mathfrak{g}^*)$ be a nonlinear LH system and 
let $\C$ be  a Casimir invariant of the Poisson algebra $S(\mathfrak{g})$, i.e., $\C=\C(v_1,\dots,v_r)$ is a polynomial function and $\{\C,v_\alpha\}=0$   for $\alpha=1,\ldots,r$, with respect to the commutation relations (\ref{71}).  Consider  the functions defined, through (\ref{74})--(\ref{77}), by
  \be
\begin{split}
\I_{\rm L}^{(k)}\!\left( h^{(k)}_{{\rm L},1},\dots, h^{(k)}_{{\rm L},r} \right)  &:=  D^{(k)}\!\left[\Delta_{\rm L}^{(k)} \bigl({\C( v_1,\dots,v_r)} \bigr) \right]\in C^\infty \bigl( {\mmm}^{k} \bigr)\subset C^\infty \bigl({\mmm}^{m+1}  \bigr), \\[4pt]
\I_{\rm R}^{(k)}\!\left( h^{(k)}_{{\rm R},1},\dots, h^{(k)}_{{\rm R},r} \right)  &:=  D^{(k)}\!\left[\Delta_{\rm R}^{(k)} \bigl({\C( v_1,\dots,v_r)} \bigr) \right]\in C^\infty \bigl( {\mmm}^{k} \bigr)\subset C^\infty \bigl({\mmm}^{m+1}  \bigr),
  \end{split}
 \label{79}
\ee
with $k=2,\dots,m+1$ and $\I^{(m+1)}\equiv \I_{\rm L}^{(m+1)}=\I_{\rm R}^{(m+1)}$. Then, $D(\C)=J^*\C$ is a constant of the motion for $X$ and if  all $\I_{\rm L}^{(k)}$ and $\I_{\rm R}^{(k)}$  are non-constant, the following is satisfied: \\
(i) Each set $\I_{\rm L}^{(k)}$ and $\I_{\rm R}^{(k)}$ is formed by $m$  $t$-independent ``left-" and ``right-constants" of the motion in $C^\infty \bigl(\mmm^{m+1} \bigr)$, which are $m$ functionally independent functions in involution with respect to the Poisson bracket (\ref{78}), thus   fulfilling the following vanishing Poisson brackets:
\be
\begin{split}
\left\{ \I_{\rm L}^{(k)}, \I_{\rm L}^{(l)}\right\}_{\Lambda^{m+1}}&= \left\{ \I_{\rm R}^{(k)}, \I_{\rm R}^{(l)}\right\}_{\Lambda^{m+1}}=0,\qquad k,l=2,\dots, m+1,\\[4pt]
 \left\{ \I_{\rm L}^{(k)}, h^{(m+1)}_{\alpha} \right\}_{\Lambda^{m+1}}&= \left\{ \I_{\rm R}^{(k)},  h^{(m+1)}_{\alpha} \right\}_{\Lambda^{m+1}}=0,\qquad  \alpha=1,\dots, r,\\[4pt]
  h^{(m+1)}_{\alpha}&=  h^{(m+1)}_{\alpha}\bigl(x_{(1)},\dots, x_{(m+1)} \bigr), 
\end{split}
 \label{710}
\ee
where  $x_{(\ell)}=\bigl\{ x^1_{(\ell)},\dots,  x^n_{(\ell)}\bigr\}$ for $\ell=1,\dots, m+1$.
\\
(ii) Both sets,  $\I_{\rm L}^{(k)}$ and $\I_{\rm R}^{(k)}$, together yield $2m-1$  
functionally independent constants of  the motion for  the  diagonal prolongation
$\widetilde X$ to the  manifold $\mm^{m+1}$, which can be understood as the nonlinear LH system 
\be
\bigl(\mm^{m+1},\Lambda^{m+1},\widetilde X,   J\otimes  {\cdots} \otimes J :\mm^{m+1}\rightarrow \mathfrak{g}^* \bigr) ,
\nonumber
\ee
giving rise to a nonlinear $t$-dependent Hamiltonian via a smooth function $F\in C^\infty(\mathbb{R}\times \mathfrak{g}^*)$ as
\be
h^{(m+1)}\bigl(t,x_{(1)},\dots, x_{(m+1)} \bigr)= F\left(t,h^{(m+1)}_{1}\bigl(x_{(1)},\dots, x_{(m+1)} \bigr),\dots, h^{(m+1)}_{r}\bigl(x_{(1)},\dots, x_{(m+1)} \bigr)\right).
 \label{712}
\ee
\end{theorem}

\begin{proof} 
For the $t$-independent case with $h^{(m+1)}\bigl(x_{(1)},\dots, x_{(m+1)} \bigr)$, the assertions have already been proven through the coalgebra formalism of superintegrable Hamiltonians in~\cite{BBR06,BHMR04,BBHMR09}, while the corresponding proofs for   linear LH systems, $h^{(m+1)} =\sum_{\alpha=1}^rb_\alpha(t)h^{(m+1)}_\alpha$, can be found in~\cite{BCHLS13,LS20,BCFHL21}. 
For the proper nonlinear situation, we apply Proposition~\ref{prop65}.   Since $\C\in C^\infty(\mathfrak{g}^*)$ is a Casimir invariant of $S(\mathfrak{g})$, then $D(\C)=J^*\C\in C^\infty(\mmm)$  is a $t$-independent constant of the motion of $X$. This can be  extended in a consistent way  by the   $k^{\rm th}$-order  tensor product  $ (J\otimes  {\cdots} \otimes J)^*\C$, leading to the constants of the motion $\I_{\rm L}^{(k)}$ and $\I_{\rm R}^{(k)}$ (satisfying  Proposition~\ref{prop63}) by means of (\ref{76}) in  $ C^\infty \bigl(\mmm^{k} \bigr)\subset  C^\infty \bigl(\mmm^{m+1} \bigr)$ for $\widetilde X$; in fact, they can be considered naturally as functions on $\mm^{m+1}$ for every $k\le   m+1$.
\end{proof}

 In addition, note that the  left- and right-constants of the motion $\I_{\rm L}^{(k)}$ and $\I_{\rm R}^{(k)}$
   allow us to obtain some other 
 constants of the motion as
\begin{equation} 
\I_{ ij} ^{(k)}=\mS_{ij} \bigl( \I^{(k)}   \bigr) , \qquad 1\le  i<j\le  m+1,\qquad k=2,\ldots,m+1,
\label{713}
\end{equation}
where $\mS_{ij}$ is the permutation of variables $x_{(i)}\leftrightarrow
x_{(j)}$ of the coordinates $\bigl(x_{(1)},\ldots,x_{(m+1)} \bigr)$ in $\mm^{m+1}$. As $\widetilde X$ is invariant under such permutations, all  $\I_{ij}^{(k)}$  are  also $t$-independent  constants of the motion  for the diagonal prolongation $\widetilde X$ to $\mm^{m+1}$. 

 It is worth observing that, given a nonlinear Hamiltonian   $h^{(m+1)}$ (\ref{712}) coming from an initial $h^{(1)}\equiv h(t,x)$,  the $(2m-1)$ functionally independent constants of the motion (\ref{710}), although they fulfill the condition (\ref{Eq:Con}),  do not lead to a proper superposition principle  (see Definition~\ref{def21})
 $$
 x_{(m+1)}=\Phii\bigl(  x_{(1)}, \ldots,x_{(m)} ; k_1,\ldots,k_n \bigr),
 $$
because $h^{(m+1)}$ is no longer a sum of $(m+1)$ copies of $h(t,x)$, as happens for the usual (linear) LH systems. In this respect, $h^{(m+1)}$ must in fact be considered as a nonlinear Hamiltonian by itself, whereas  
 its corresponding constants of the motion  should be regarded as  integrability conditions among the 
variables within $h^{(m+1)}$.

 %%%%%%%%%%%%%%%%%%%%%%%%%%%%%%%%%%%%%%%%%%%

\section{Nonlinear Lie--Hamilton systems with   Poisson \texorpdfstring{$\mathfrak{sl}(2,\mathbb{R})$}{sl(2,R)}-coalgebra symmetry}
\label{s8}

The formalism of nonlinear LH systems provides a large number of applications, due to the freedom in the choice of a Poisson coalgebra $(\mathfrak{g},\Delta )$, as well as of functions $F\in C^\infty(\mathbb{R}\times \mathfrak{g}^*)$. In order to illustrate Theorem~\ref{teor7}, and as motivation for applications that will be developed in the next section, we consider the Poisson $\mathfrak{sl}(2,\mathbb{R})$-coalgebra spanned by the basis $\{ v_-,v_+,v_3\}$ with Poisson brackets
\be
\bigl\{v_3,v_+ \bigr\} = 2 v_+,\qquad \bigl\{v_3,v_- \bigr\} =- 2 v_-, \qquad\bigl\{v_-,v_+ \bigr\} = 4 v_3,
\label{w1}
\ee
Casimir invariant
\be
\C =v_-v_+-v_3^2,
\label{w2}
\ee
and  primitive coproduct given by (\ref{72}).   Hereafter, we will assume the notation in Section~\ref{s21} by taking canonical position and momenta coordinates, $\>q:=(q_1,\dots,q_n)$  and  $\>p:=(p_1,\dots,p_n)$, defined on the manifold ${\rm T}^*\mathbb{R}^n_0:=\{(\>q,\>p)\in {\rm T}^*\mathbb{R}^n:   \>\prod_{i=1}^n|q_i|^{|c_i|}\neq 0\}$, and the Poisson bivector (\ref{bivector}) inducing the canonical Poisson bracket,  given by
\be
\{f,g\}_\Lambda=\sum_{i=1}^n\left(\frac{\partial f}{\partial q_i}
\frac{\partial g}{\partial p_i}
-\frac{\partial g}{\partial q_i} 
\frac{\partial f}{\partial p_i}\right),\qquad \forall f,g\in C^\infty\bigl({\rm T}^*\mathbb{R}^n_0 \bigr).
\label{w3}
\ee
Then, by using (\ref{77}) with $(m+1)\equiv n$, we construct the following Hamiltonian functions on $C^\infty\bigl({\rm T}^*\mathbb{R}^n_0 \bigr)$:
  \be
\begin{split}
h^{(1)}_-&:=D( v_-)= q_1^2  , \qquad   h^{(1)}_+:=D( v_+)= p_1^2 +\frac{c_1}{q_1^2} , \qquad h^{(1)}_3:=D( v_3)= q_1p_1  , \qquad \\ 
 h^{(k)}_{{\rm L},-}&:=   \sum_{i =1}^k q_i^2  ,\qquad h^{(k)}_{{\rm L},+}:=   \sum_{i =1}^k\left( p_i^2+\frac{c_i}{q_i^2} \right),\qquad h^{(k)}_{{\rm L},3}:=   \sum_{i =1}^k q_i p_i ,\\[2pt]
  h^{(k)}_{{\rm R},-}&:= \sum_{i =n-k+1}^{n}  q_i^2  ,\qquad h^{(k)}_{{\rm R},+}:=   \sum_{i =n-k+1}^{n}  \left(  p_i^2+\frac{c_i}{q_i^2} \right),\qquad h^{(k)}_{{\rm R},3}:=    \sum_{i =n-k+1}^{n}  q_i p_i ,
 \end{split}
 \label{w4}
\ee
with $k=2,\dots,n$, and where $c_i$ are arbitrary real constants. As a shorthand notation, let us denote
\be
\>q^2:=  \sum_{i =1}^n q_i^2  ,\qquad \>p^2:=  \sum_{i =1}^n p_i^2  ,\qquad \tilde{ \>p}^2:=   \>p^2+\sum_{i =1}^n   \frac{c_i}{q_i^2},\qquad \>q\scal \>p:=  \sum_{i =1}^n q_i p_i ,
 \label{w5}
\ee
such that the functions (\ref{w4}) for $k=n$ are simply written as
\be
h^{(n)}_{-}= \>q^2,\qquad h^{(n)}_{+}=  \tilde{ \>p}^2  , \qquad h^{(n)}_{3}=  \>q\scal  \>p ,
 \label{w6}
\ee
fulfilling the same Poisson brackets (\ref{w1}) with respect to (\ref{w3}). Then, the most general nonlinear LH system (\ref{712})  with such  Poisson $\mathfrak{sl}(2,\mathbb{R})$-coalgebra symmetry turns out to be
\be
h^{(n)}\bigl(t,  \>q,  \>p \bigr)= F\!\left(t,h^{(n)}_{-} ,h^{(n)}_{+},h^{(n)}_{3}\right)= F\!\left(t, \>q^2 ,\tilde{ \>p}^2 , \>q\scal  \>p\right), \qquad \forall ( \>q, \>p)\in {\rm T}^*\mathbb{R}^n_0,
 \label{w7}
\ee
for any  smooth function $F\in C^\infty(\mathbb{R}\times \mathfrak{sl}(2,\mathbb{R})^*)$.
 
 The fundamental vector fields associated with the Hamiltonian functions (\ref{w6}) follow from the relation (\ref{conditions}), namely 
 \begin{equation}
X_-= -2\sum_{i=1}^n q_i\frac{\partial}{\partial p_i},\qquad 
X_+=2\sum_{i=1}^n\left(p_i\frac{\partial}{\partial
q_i}+\frac{c_i}{q_i^3}\frac{\partial}{\partial p_i}\right),\qquad X_3= \sum_{i=1}^n\left(q_i\frac{\partial}{\partial
q_i}-p_i\frac{\partial}{\partial p_i}\right),
\label{w8}
\end{equation}
which satisfy the following commutators
\be
[X_3,X_+]=-2X_+,\qquad [X_3,X_-]= 2X_-,\qquad [X_-,X_+]=-4X_3.
\nonumber
\ee
 Therefore, the $t$-dependent   vector field provided by (\ref{w7}) is obtained by applying 
  (\ref{5xx}), leading to
 \be
X(t,\>q,\>p)=  \frac{\partial F}{\partial h_-^{(n)} } \, X_{-}+\frac{\partial F}{\partial h_+^{(n)} } \, X_{+}
+\frac{\partial F}{\partial h_3^{(n)} } \, X_{3} ,\qquad \forall t\in \mathbb{R},\qquad \forall (\>q,\>p)\in {\rm T}^*\mathbb{R}^n_0.
\label{w9}
\ee

From Theorem~\ref{teor7}, the $t$-independent constants of the motion (\ref{79}) for $h^{(n)}$ (\ref{w7}), coming from the Casimir invariant (\ref{w2}),  are found to be~\cite{BHMR04}  
 \be
 \begin{split} 
  \I_{\rm L}^{(k)}&     = h^{(k)}_{{\rm L},-}  h^{(k)}_{{\rm L},+}-  \left( h^{(k)}_{{\rm L},3}  \right)^2\\[2pt]
    &= \sum_{1\leq i<j}^k \left( ({q_i}{p_j} -
{q_j}{p_i})^2 +  
c_i\frac{q_j^2}{q_i^2}+c_j\frac{q_i^2}{q_j^2} \right)
+\sum_{i=1}^k c_i  ,\\[2pt]
  \I_{\rm R}^{(k)}&     = h^{(k)}_{{\rm R},-}  h^{(k)}_{{\rm R},+}-  \left( h^{(k)}_{{\rm R},3}  \right)^2\\[2pt]
  &= \sum_{n-k+1\leq i<j}^n \left( ({q_i}{p_j} -
{q_j}{p_i})^2 +  
c_i\frac{q_j^2}{q_i^2}+c_j\frac{q_i^2}{q_j^2} \right)
+\!\sum_{i=n-k+1}^n\! c_i ,
 \end{split}
\label{w10}
\ee
with $k=2,\dots, n$, while $I^{(n)}_{L}=I^{(n)}_{R}$ and
\be
  \I^{(1)}      = h^{(1)}_{-}  h^{(1)}_{+}-  \left( h^{(1)}_{3}  \right)^2 = c_1,
\label{I1}
\ee
 which is the constant that labels the representation of $ \mathfrak{sl}(2,\mathbb{R})^\ast$. Thus,  the  value of $c_1$ yields   three types  of submanifolds corresponding to the surfaces with constant value   of the Casimir $\mathcal{C}$ (\ref{w2}) for the Poisson structure of $ \mathfrak{sl}(2,\mathbb{R})$. Moreover, this determines the three non-diffeomorphic classes of $\mathfrak{sl}(2,\mathbb{R})$-LH systems in the classification of LH systems of the plane~\cite{LS20}.  Observe that, in fact, for an $i$-th copy of $S( \mathfrak{sl}(2,\mathbb{R})^\ast)$ within the  $n^{\rm th}$-order tensor product $S( \mathfrak{sl}(2,\mathbb{R})^\ast)\otimes\dots \otimes S( \mathfrak{sl}(2,\mathbb{R})^\ast)$ it follows that
\be
  \I^{(i)}      = h^{(i)}_{-}  h^{(i)}_{+}-  \left( h^{(i)}_{3}  \right)^2 = c_i,\qquad i=1,\dots,n .
\label{wwx}
\ee
Therefore, the nonlinear $ \mathfrak{sl}(2,\mathbb{R})$-Hamiltonian  (\ref{w7}) is constructed 
 over a mixture of such representations of $ \mathfrak{sl}(2,\mathbb{R})^\ast$ that only lead to a single one whenever all $c_i=c$.

The  expressions (\ref{w10}) can be better understood as the sum of ``generalized" angular momentum components defined by
\be
 \mL_{ij}:=({q_i}{p_j} -
{q_j}{p_i})^2 +  
c_i\frac{q_j^2}{q_i^2}+c_j\frac{q_i^2}{q_j^2} ,\qquad 1\le i<j\le n,
\label{LLa}
 \ee
 in such a manner that when all $c_i=0$, they reduce to the square of the usual  angular momentum functions $\mJ_{ij}$ in the form
 \be
  \mL_{ij}=\mJ_{ij}^2,\qquad  \mJ_{ij}:=q_{i}{p_j} -
{q_j}{p_i} ,
\label{Lb}
  \ee
 the latter closing on a  Lie--Poisson algebra isomorphic to  $\mathfrak{so}(n)$. If, for instance, we take $n=4$, by 
 Theorem~\ref{teor7} we get $5=(2m-1)=(2n-3)$ $t$-independent constants of the motion (\ref{w10}), which are functionally independent, and explicitly given by 
 \be
 \begin{split}
  \I_{\rm L}^{(2)}&= \mL_{12}+c_1+c_2 ,\qquad  \I_{\rm L}^{(3)}= \mL_{12}+ \mL_{13}+ \mL_{23}+c_1+c_2+c_3,\\[2pt]
   \I_{\rm R}^{(2)}&= \mL_{34}+c_3+c_4 ,\qquad  \I_{\rm R}^{(3)}= \mL_{23}+ \mL_{24}+ \mL_{34}+c_2+c_3+c_4,\\[2pt]
     \I^{(4)}&\equiv \I_{\rm L}^{(4)}\equiv \I_{\rm R}^{(4)}=\mL_{12}+ \mL_{13}+ \mL_{14}+  \mL_{23}+ \mL_{24}+ \mL_{34}+c_1+c_2+c_3+c_4 .
    \end{split}
\nonumber
\ee
 thus providing  integrability relations among the $(4+4)$ variables $(\>q(t),\>p(t))$.   Observe that by permutation of indices (\ref{713}), other $t$-independent constants of the motion can be deduced, although  they are not  functionally independent with respect to all the previous five, e.g.,
\be
\mS_{13} \!\left(   \I_{\rm L}^{(2)}  \right) = \mL_{23}+c_2+c_3,\qquad
 \mS_{13} \! \left(   \I_{\rm R}^{(2)}  \right) =\mL_{14}+c_1+c_4.
 \nonumber
\ee
 
 Summing up, {\em any} $t$-dependent Hamiltonian $h^{(n)}\bigl(t,\>q, \>p \bigr)$ of type (\ref{w7}), determined by a specific function $F$, is endowed with the $(2n-3)$ $t$-independent constants of the motion (\ref{w10}) and, in this sense, they can be regarded as its ``universal" constants of the motion~\cite{BH07}.  Recall  also that  the Racah algebra for $\mathfrak{sl}(2,\mathbb R)$, with such constants of the motion,  has been  studied in detail in~\cite{Latini2019,Latini2021,Latini2021b}.

  %%%%%%%%%%%%%%%%%%%%%%%%%%%%%%%%%%%%%%%%%%%

\subsection{Generalized {\em t}-dependent central potentials on \texorpdfstring{$\mathfrak{sl}(2,\mathbb{R})$}{sl(2,R)}-coalgebra spaces}
\label{s81}

 The  Poisson $\mathfrak{sl}(2,\mathbb{R})$-coalgebra allows us to construct directly $t$-dependent central potentials on Riemannian spaces of (generically) non-constant curvature. Therefore, among others, $t$-dependent  (curved) oscillators and Kepler--Coulomb (KC) potentials emerge as particularly relevant systems of this type, which will be addressed   in  the following section, thus extending already known $t$-independent Hamiltonians  to this framework. 
 
Let us consider a Hamiltonian in classical mechanics expressed in terms of a $t$-independent kinetic energy $\mT$ and a $t$-dependent potential $\mU$ on a Riemannian $n$-dimensional manifold $\mm$. For the sake of simplicity,   we will omit the index ``$(n)$" from now on, as we will always deal with  $n$ dimensions, unless otherwise stated. We start with the most general expression of the kinetic energy 
for the  Hamiltonian $h$ (\ref{w7}) which, by setting all constants $c_i=0$ in the Hamiltonian function $ h_{+}$ (\ref{w6}),  reads as~\cite{sl2spaces}
\be
\mT(\>q,\>p)= \mA  (h_{-}  )h_{+}+\mB  (h_{-}  ) h_{3} ^2=  \mA  \bigl(\>q^2 \bigr)\>p^2+\mB  \bigl(\>q^2 \bigr) ( \>q\scal  \>p )^2 ,
\label{za}
\ee
where $\mA$ and $\mB$ are arbitrary functions which determine the specific underlying manifold $\mm$. It is worth recalling that $\mT$  is associated with the so-called $\mathfrak{sl}(2,\mathbb{R})$-coalgebra spaces~\cite{sl2spaces,annals2009}  which include, as a remarkable class, the spherically symmetric spaces with metric given by
\be
\dd s^2=\ff(|\>q|)^2\dd\>q^2,\qquad |\>q|:=\sqrt{\>q^2},\qquad \dd\>q^2:=\sum_{i=1}^n \dd q_i^2 ,
\label{zb}
\ee
where $\ff(|\>q|)$ is an arbitrary smooth function playing the role of  a conformal factor of the Euclidean metric $\dd\>q^2$; in hyperspherical coordinates, the variable $|\>q|$ is just a radial coordinate $r$. The scalar curvature $R$ of the metric  (\ref{zb}) is determined by the conformal factor and the dimension of the manifold, turning out to be 
\be
R(|\bq|)=-(n-1)\,\frac{    (n-4)\ff '(|\bq|)^2+ \ff(|\bq|)  \left(    2\ff''(|\bq|)+2(n-1)|\bq|^{-1}\ff'(|\bq|)  \right)}   {\ff(|\bq|)^4  }   .
\label{zc}
\ee
Then, the metric  (\ref{zb})  leads to a free Hamiltonian  $\mT$ (\ref{za}) describing the geodesic motion in  a spherically symmetric space $\mm$ by choosing  
\be
 \mA  (h_{-}  )=\frac{1}{2 \ff\bigl(\!\sqrt{h_{-} } \,\bigr)^2},\qquad \mB  (h_{-}  )=0,
 \nonumber
 \ee
 so that
 \be
 \mT(\>q,\>p)=\frac{h_{+}}{2\ff\bigl(\!\sqrt{h_{-} }\, \bigr)^2 }=\frac{\>p^2}{2 \ff(|\>q|)^2}.
 \label{Th}
 \ee
From this, a $t$-dependent central potential $\mU$ can be added, obtaining a family of $t$-dependent Hamiltonians given by
\be
h  (t,  \>q,  \>p  )=\frac{h_{+}}{2\ff\bigl(\!\sqrt{h_{-} }\, \bigr)^2 }+\mU\bigl(t, \! \sqrt{h_{-} }\, \bigr)
=\frac{\>p^2}{2 \ff(|\>q|)^2}+\mU (t, |\>q|    ),
\nonumber
\ee
 which entails, among others, $t$-dependent curved oscillators and KC systems.
A straightforward generalization follows by keeping arbitrary constants $c_i$ in $h_+$, yielding
\be
\begin{split}
h  (t,  \>q,  \>p  )&= F\!\left(t,h_{-} ,h_{+},h_{3}\right)=\frac{h_{+}}{2\ff\bigl(\!\sqrt{h_{-} }\, \bigr)^2 }+\mU\bigl(t,\!\sqrt{h_{-} }\, \bigr)\\
&=\frac{\>p^2}{2 \ff(|\>q|)^2}+\mU (t, |\>q|  )+\frac{1}{{2 \ff(|\>q|)^2}}\sum_{i=1}^n \frac{c_i}{q_i^2},
\end{split}
\label{gH}
\ee
which  provides a  Hamiltonian describing the motion of  a particle with unit mass under the superposition of a   $t$-dependent central potential with $n$ arbitrary Rosochatius--Winternitz potentials  in a spherically symmetric space.
 Therefore, from this construction, the appearance of such ``curved" Rosochatius--Winternitz potentials is a direct consequence of the choice of $h_+$ in the kinetic term, which eventually relies on the representation of the   $i$-th copy of $S( \mathfrak{sl}(2,\mathbb{R})^\ast)$ determined by each $c_i$ (\ref{wwx}) \cite{annals2009}.
 It is worthy to stress that the Hamiltonian (\ref{gH}) admits an alternative relevant physical interpretation as 
 a position-dependent mass (PDM) system~\cite{Plastino} with $t$-dependent potential by simply identifying the square of the conformal factor with   a radial PDM, $\ff(|\>q|)^2\equiv m(|\>q|)$, so that the PDM naturally arises in the Rosochatius--Winternitz potentials; explicitly  
 \be
 h  (t,  \>q,  \>p  )=\frac{\>p^2}{2 m(|\>q|)}+\mU (t, |\>q|   )+\frac{1}{{2 m(|\>q|)}}\sum_{i=1}^n \frac{c_i}{q_i^2}.
 \label{PDMg}
 \ee
Recall that $t$-independent PDM systems have been widely considered for classical and quantum oscillators~\cite{CrNN07,Quesne07,CrR09,BurgosAnnPh11,GhoshRoy15,Quesne15Jmp} as well as for KC systems~\cite{QuesneKC}.

 The Hamilton equations of  $h  (t,  \>q,  \>p  )$ (\ref{gH})  lead to the system of differential equations given by
\begin{equation}
\left\{\begin{aligned}
 \frac{\dd q_i}{\dd t}&=\frac{p_i}{ \ff(|\>q|)^2 }, \qquad i=1,\ldots,n,\\[2pt]
  \frac{\dd p_i}{\dd t}&=- \frac{q_i}{|\>q|}\,\frac{\partial\, \mU (t, |\>q|    )}{\partial |\>q|} + \frac{q_i}{|\>q|}\,
  \frac{\ff'(|\>q|)}{ \ff(|\>q|)^3 }\left(\>p^2+\sum_{i=1}^n \frac{c_i }{q_i^2} \right)+\frac{c_i }{ q_i^3 \ff(|\>q|)^2}.
  \end{aligned}\right.  
\label{eqs}
\end{equation}
The associated $t$-dependent vector field is then obtained by using (\ref{w9}) with the function $F$    in (\ref{gH}):
\be
 \begin{aligned}
 X&= \frac{1}{2 \ff\bigl(\!\sqrt{h_{-} } \,\bigr)^2 }\, X_+ +
 \left(  \frac{1}{2\sqrt{h_{-}  }} \,\frac{\partial\, \mU \bigl(t,\!\sqrt{h_{-} } \, \bigr) }{\partial \sqrt{h_{-} }  } -     \frac{h_+}{ 2\sqrt{h_{-} } \, \ff\bigl(\!\sqrt{h_{-} } \,\bigr)^3}\,
\frac{\dd \ff\bigl(\!\sqrt{h_{-} } \,\bigr)}{\dd    \sqrt{h_{-} } } \right)\! X_-  
 \\[4pt]
&= \sum_{i=1}^n \left[ \frac {p_i}{ \ff(|\>q|)^2 } \, \frac{\partial}{\partial q_i}- \left(  \frac{q_i}{|\>q|}\,\frac{\partial\, \mU (t, |\>q|    )}{\partial |\>q|} - \frac{q_i}{|\>q|}\,
  \frac{\ff'(|\>q|)}{ \ff(|\>q|)^3 }\left(\>p^2+\sum_{i=1}^n \frac{c_i}{q_i^2} \right)-\frac{c_i}{ q_i^3 \ff(|\>q|)^2}
  \right)\!\frac{\partial}{\partial p_i} \right],
 \end{aligned}
\label{gX}
\ee
which, consistently, gives rise to the system of equations (\ref{eqs}). 
  
 Consequently, the  $t$-dependent nonlinear Hamiltonian  $h$ (\ref{gH}), or its associated $t$-dependent vector field $X$ (\ref{gX}), is formed by a ``main" $t$-dependent potential $\mU$ plus $n$ arbitrary  $t$-independent ``centrifugal" terms.

  \newpage
 %%%%%%%%%%%%%%%%%%%%%%%%%%%%%%%%%%%%%%%%%%%%%%

\section{Applications} 
\label{s9}

We present some specific relevant  $t$-dependent ``generalized" Hamiltonians systems,  for which   the terminology ``generalized" refers to the addition of   $t$-independent  $n$ Rosochatius--Winternitz  terms to a $t$-dependent main central potential (\ref{gH}). 
    We remark that such applications  correspond to curved oscillators in Section~\ref{s91} and KC systems in Section~\ref{s92}, which  are just  the $t$-dependent counterpart of well-known $t$-independent maximally superintegrable curved systems, i.e., all of them are endowed, besides the $(2n-3)$ ``universal"  $t$-independent constants  of the motion  (\ref{w10}), with   a Demkov--Fradkin tensor for oscillators or a Laplace--Runge--Lenz $n$-vector for KC systems.

 %%%%%%%%%%%%%%%%%%%%%%%%%%%%%%%%%%%%%%%%%%%%%%
 
\subsection{Generalized curved oscillators with {\em t}-dependent frequency} 
\label{s91}

We focus here on three types of $n$-dimensional  generalized  oscillator potentials with  $t$-dependent frequency $\omega(t)$, for which their  Hamiltonian (\ref{gH}),   equations of the motion (\ref{eqs}) together with their associated vector field (\ref{gX}) are  explicitly displayed in Table~\ref{table2}.  Let us   describe, separately,  their main features and   relationships with the literature.

 %%%%%%%%%%%%%%%%%%%%%%%%%%%%%%%%%%%%%%%%%%%%%%

\subsubsection{Generalized Euclidean isotropic  oscillator} 
\label{s911}

For completeness and to explain the role of $c_i$-potentials in detail, let us first  consider  the
isotropic   oscillator with $t$-dependent frequency $\omega(t)$ in Section~\ref{s61}. The main potential in  (\ref{gH})  is  
\be
\mU(t,|\>q|)=\frac 12 \omega^2({t}) h_- =\frac 12 \omega^2({t})\>q^2 ,
\label{Usw}
\ee
 which is endowed with the $t$-independent angular momentum constants of the motion (\ref{Lb}). Recall that $t$-dependent harmonic oscillators have been applied a wide variety of contexts  (see, e.g.,~\cite{LewisHR,LiLi,Guasti,Soliani,Kovacic,Robnik,Prencel} and references therein).
    The linear LH system is defined on the Euclidean space $\bE^n\equiv {\rm T}^*\mathbb{R}^n$ expressed in Cartesian variables $(\>p,\>q)$.    Then, we add the $n$ arbitrary ``centrifugal" $c_i$-potentials 
coming from $\tilde{ \>p}^2$ instead of $ { \>p}^2$  (\ref{w5}), arriving at the generalized  Euclidean isotropic  oscillator or SW system described in Section~\ref{s21}, which is now defined on 
$\bE^n_0\equiv {\rm T}^*\mathbb{R}^n_0$ with $\>q\ne 0$, as summarized in Table~\ref{table2}.
Observe that the correspondence between the expressions in Section~\ref{s21} and those in 
 Section~\ref{s8}, is provided by 
  $ h_1 \to \tfrac12 h_-$,  $h_2 \to - \tfrac12 h_3$ and $h_3 \to \tfrac12 h_+$, thus from Table~\ref{table2} it is found that $h =h_{\rm SW}$ (\ref{hhSW}) and $X =X_{\rm SW}$  (\ref{xsw}).

Let us recover some particular systems. If we  set all $c_i= c$, we obtain the one-dimensional SW system on $\bE_0\equiv {\rm T}^*\mathbb{R}_0$ with variables $(q_1\ne 0,p_1)$, i.e., equivalently,   the Ermakov system (\ref{Ea})~\cite{Ermakov1880,Leach1991,Maamache1995,LS20,Leach2008} or the Milne--Pinney  equation~\cite{Milne1930,Pinney1950}. Then, its $n$-dimensional  version corresponds to  the diagonal prolongation  to $\bE^n_0\equiv {\rm T}^*\mathbb{R}^n_0$. The corresponding superposition principle follows by fixing $(m+1)=3$ in (\ref{Eq:Con}) with the $t$-independent constants of the motion (\ref{w10}) reading as $  \I^{(1)}  =c$ (\ref{I1}) (remind that $c
\in \mathbb{R}$ labels the three non-diffeomorphic classes of    LH systems on the plane~\cite{LS20}),  and the others are expressed on $\bE^3_0\equiv {\rm T}^*\mathbb{R}^3_0$ (with coordinates $(q_i\ne 0,p_i)$ and $i=1,2,3$) as
\be
 \begin{split}
  \I_{\rm L}^{(2)}&=({q_1}{p_2} -
{q_2}{p_1})^2 +  
c\left( \frac{q_2^2}{q_1^2}+  \frac{q_1^2}{q_2^2}\right) +2c=  \mL_{12}  +2c , \\[2pt]
   \I_{\rm R}^{(2)}&= ({q_2}{p_3} -
{q_3}{p_2})^2 +  
c\left( \frac{q_3^2}{q_2^2}+  \frac{q_2^2}{q_3^2}\right) +2c=  \mL_{23} +2c  ,\\[2pt]
     \I^{(3)}&\equiv \I_{\rm L}^{(3)}\equiv \I_{\rm R}^{(3)}= 
       \mL_{12}+ \mL_{13}+  \mL_{23}+ 3c  ,
    \end{split}
\label{z1}
\ee
  with $\mL_{ij}$ given in (\ref{LLa}).  By permutations of the indices (\ref{713})  in $\I_{\rm L}^{(2)}$ or $  \I_{\rm R}^{(2)}$, another constant is found:
$$
 \I_{\rm 13}^{(2)}  =\mS_{23} \!\left(   \I_{\rm L}^{(2)}  \right) = \mS_{12} \!\left(   \I_{\rm R}^{(2)}  \right) =   ({q_1}{p_3} -
{q_3}{p_1})^2 +  
c\left( \frac{q_3^2}{q_1^2}+  \frac{q_1^2}{q_3^2}\right) +2c=\mL_{13}+2c    ,  
$$
such that
$  \I^{(3)} =
      \I_{\rm L}^{(2)}+ \I_{\rm R}^{(2)}+ \I_{\rm 13}^{(2)}-3c .
$
  To obtain  a superposition principle it is necessary, in this case, to consider only two functionally independent constants of the motion, which can be taken in different ways. For instance, if we select $ \I_{\rm L}^{(2)}\equiv k_1$ and $ \I_{\rm 13}^{(2)}\equiv k_2 $, fulfilling (\ref{Eq:Con}), since $(m+1)=3$, the general solution $(q_1(t),p_1(t))$ can be expressed in terms of two particular solutions $(q_2,p_2,q_3,p_3)$ and   two constants $k_1$ and $k_2$, as done in~\cite{BCHLS13}.
 Therefore, the functions  (\ref{w10}), together with the permutations of their indices (\ref{713}), yield  a large number constants of the motion (not all of them functionally independent), but only a few are required to deduce a superposition rule. 
     
   Another well-known particular case is recovered by choosing $c_1=c$ and all the remaining $c_j=0$ $(j=2,\dots,n)$ in the  space  $\bE_{q_1\ne 0}^n\equiv {\rm T}^*\mathbb{R}_{q_1\ne 0}^n$,  thus extending the previous Ermarkov system, namely~\cite{Leach1991,LS20}
\begin{equation} 
\left\{\begin{aligned}
\frac{\dd q_1}{\dd t}&=p_1, \qquad \frac{\dd p_1}{\dd t} =-\omega^2(t)q_1+\frac{c}{q_1^3},\\[2pt]
\frac{\dd q_j}{\dd t}&=p_j , \qquad \frac{\dd p_j}{\dd t} =-\omega^2(t)q_j,\qquad j=2,\dots, n .
\end{aligned}\right. 
\nonumber
\end{equation}
For $n=2$, the set of constants of the motion  (\ref{w10}) lead to a single constant for the Hamiltonian  $h^{(2)} $ (\ref{gH}) in  $\bE_{q_1\ne 0}^2$, 
\be
   \I^{(2)} \equiv \I_{\rm L}^{(2)}\equiv \I_{\rm R}^{(2)}=({q_1}{p_2} -
{q_2}{p_1})^2 +  
c\left( 1+\frac{q_2^2}{q_1^2} \right)  ,
\nonumber
\ee
corresponding to the Ermakov--Lewis invariant~\cite{Leach2008}.  If $n=3$, the functions (\ref{w10}) give  rise to three constants of the motion, similarly to (\ref{z1}), reading now as
\be
 \begin{split}
  \I_{\rm L}^{(2)}&=({q_1}{p_2} -
{q_2}{p_1})^2 +  
c\left( 1+\frac{q_2^2}{q_1^2} \right),\qquad
  \I_{\rm R}^{(2)} =({q_2}{p_3} -
{q_3}{p_2})^2  ,
 \\[2pt]
     \I^{(3)}&\equiv \I_{\rm L}^{(3)}\equiv \I_{\rm R}^{(3)}=\sum_{1\le i<j}^3
   ({q_i}{p_j} -
{q_j}{p_i})^2+  
c \,\frac{q_2^2}{q_1^2} +  
c \,\frac{q_3^2}{q_1^2}+c.
    \end{split}
\nonumber
\ee
If we consider
$$
 \I_{\rm 13}^{(2)}  =\mS_{23} \!\left(   \I_{\rm L}^{(2)}  \right)=({q_1}{p_3} -
{q_3}{p_1})^2 +  
c\left( 1+\frac{q_3^2}{q_1^2} \right),
$$
we obtain another Ermakov--Lewis invariant such that
  $ \I^{(3)}= \I_{\rm L}^{(2)}+ \I_{\rm 13}^{(2)} + \I_{\rm R}^{(2)} -c$, which is formed by two Ermakov--Lewis invariants and the square of the component of the  angular momentum in the plane $q_2q_3$, thus recovering the results presented in~\cite{LS20}. Furthermore,  we recall that the consideration of a single $c_1$-potential corresponds to the presence of a centrifugal (infinite) barrier which limits the trajectory of the particle on the Euclidean space; for $n=2$ it is half of the Euclidean plane   and if a second potential with $c_2\ne 0$ is added, the trajectory is restricted to a quadrant.
In this sense, the  SW system is also called ``caged" oscillator~\cite{VerrierOscillator}.

Other particular cases can be constructed on the Euclidean space, keeping the same main potential (\ref{Usw}) for other  choices of the $c_i$-terms.
 
  %%%%%%%%%%%%%%%% TABLE 2%%%%%%%%%%%%%%%%%%
 \FloatBarrier
   
 \begin{table}[ht!]
{\small
\caption{ \small{$t$-Dependent  generalized curved oscillators  from the Poisson $\mathfrak{sl}(2,\mathbb{R})$-coalgebra in Section~\ref{s8}, 
 with Hamiltonian functions $h_\alpha$ (\ref{w6}) and vector fields $X_\alpha$ (\ref{w8})  $(\alpha=-,+,3)$.  For each case, we indicate the $t$-dependent Hamiltonian, equations of the motion and vector field; for all, $\>\prod_{i=1}^n|q_i|^{|c_i|}\ne 0$ and  $i=1,\dots,n$.}}
  \begin{center}
\noindent 
\begin{tabular}{l}
\hline

\hline
\\[-6pt]
$\bullet$ Generalized Euclidean isotropic oscillator or SW system:  \quad  $\ff =1 $\qquad $\mU = \frac 12 \omega^2(t)h_-$ 
\\[4pt]
   $\displaystyle{\quad h =\frac 12h_+ +\frac 12 \omega^2(t)h_- =  \frac{1}{2}\>p^2 + \frac 12 \omega^2(t)\>q^2 +\frac 12\sum_{i=1}^n \frac{c_i}{q_i^2} }$
   \\[8pt]
 $\displaystyle{\quad   \frac{\dd q_i}{\dd t}= p_i \qquad   \frac{\dd p_i}{\dd t}=- \omega^2(t)q_i +\frac{c_i}{q_i^3} }$\\[10pt]
   $\displaystyle{\quad X =\frac 12X_+ +\frac 12 \omega^2(t)X_-   }$
  \\[10pt]
  \hline
\\[-6pt]
 
$\bullet$ Generalized spherical and hyperbolic isotropic oscillators\\[6pt]
  {\em (i) Poincar\'e variables} (curvature $\kappa$):   \quad  $\ff =(1+\kappa h_-)^{-1} $\qquad $\mU = \frac 12 \omega^2(t)h_-(1-\kappa h_-)^{-2}$
\\[4pt]
    $\displaystyle{\quad h =\frac 12(1+\kappa h_-)^{2}h_+ +\frac 12 \omega^2(t)\frac{h_-}{(1-\kappa h_-)^{2}}   }$
   \\[8pt]
      $\displaystyle{\qquad\!  =  \frac{1}{2}(1+\kappa \>q^2)^{2}\>p^2 + \frac 12 \omega^2(t)\frac{\>q^2}{(1-\kappa \>q^2)^{2}}   + \frac{1}{2}(1+\kappa \>q^2)^{2}\sum_{i=1}^n \frac{c_i}{q_i^2} }$
   \\[10pt]
 $\displaystyle{\quad   \frac{\dd q_i}{\dd t}=(1+\kappa \>q^2)^{2} p_i \qquad   \frac{\dd p_i}{\dd t}=- \omega^2(t)q_i \frac{(1+\kappa \>q^2)}{(1-\kappa \>q^2)^3} +\frac{c_i(1+\kappa \>q^2)^{2}}{q_i^3} - 2\kappa q_i (1+\kappa \>q^2) \!\left( \!\>p^2+\sum_{i=1}^n  \frac{c_i}{q_i^2} \right)  }$
 \\[12pt]
   $\displaystyle{\quad X =\frac 12(1+\kappa h_-)^{2}X_+ + \left(\frac 12 \omega^2(t) \frac  {  (1+\kappa h_-) }{  (1-\kappa h_-)^3}+\kappa (1+\kappa h_-)h_+ \!\right)\!X_-  }$
    \\[12pt]
   
     {\em  (ii) Beltrami variables} (curvature $\kappa$):   \quad  $\mA =\frac12(1+\kappa h_-)  $\qquad$\mB =\frac12 \kappa(1+\kappa h_-)  $\qquad $\mU = \frac 12 \omega^2(t)h_- $
   \\[4pt]
    $\displaystyle{\quad h =\frac 12(1+\kappa h_-) \bigl( h_+ + \kappa h_3^2 \bigr) +\frac 12 \omega^2(t) {h_-}  }$
   \\[8pt]
      $\displaystyle{\qquad\!  =  \frac{1}{2} (1+\kappa \>q^2 ) \bigl(\>p^2 + \kappa  ( \>q\scal  \>p )^2 \bigr)+ \frac 12 \omega^2(t) {\>q^2}   + \frac{1}{2}(1+\kappa \>q^2) \sum_{i=1}^n \frac{c_i}{q_i^2} }$
   \\[8pt]
 $\displaystyle{\quad   \frac{\dd q_i}{\dd t}=(1+\kappa \>q^2) \bigl( p_i  + \kappa q_i ( \>q\scal  \>p )\bigr)   }$
 \\[8pt]
 $\displaystyle{\quad      \frac{\dd p_i}{\dd t}=- \omega^2(t)q_i   +(1+\kappa \>q^2)\left(\! \frac{c_i }{q_i^3}-\kappa p_i  ( \>q\scal  \>p )\! \right)-  \kappa q_i  \!\left(   \!\>p^2+ \kappa ( \>q\scal  \>p )^2 + \sum_{i=1}^n  \frac{c_i}{q_i^2} \right)  }$
 \\[10pt]
   $\displaystyle{\quad X =\frac 12(1+\kappa h_-) X_+ +  \kappa (1+\kappa h_-) h_3 X_3 +\frac 12\! \left( \omega^2(t)  +\kappa \bigr(h_+ +\kappa h_3^2 \bigl) \right)\! X_-  }$
       \\[10pt]
  \hline
\\[-5pt]

$\bullet$ Generalized Darboux III oscillator:   \quad  $\ff =\sqrt{1+\lambda h_- }$\qquad $\mU = \frac 12 \omega^2(t)h_-(1+\lambda h_-)^{-1}$ 
\\[6pt]
   $\displaystyle{\quad h =\frac  {h_+}{2(1+\lambda h_-)} +  \omega^2(t)\frac  {h_-}{2(1+\lambda h_-)}  =  \frac{1}{2(1+\lambda \>q^2)}\left(\!\>p^2 +  \omega^2(t)\>q^2 +\sum_{i=1}^n \frac{c_i}{q_i^2} \right)}$
   \\[14pt]
 $\displaystyle{\quad   \frac{\dd q_i}{\dd t}=\frac {p_i }{1+\lambda \>q^2}\qquad   \frac{\dd p_i}{\dd t}=- \omega^2(t)\frac{q_i}{(1+\lambda \>q^2)^2} +\frac{c_i}{q_i^3(1+\lambda \>q^2)} + \frac{\lambda q_i}{(1+\lambda \>q^2)^2}    \!\left( \!\>p^2+\sum_{i=1}^n  \frac{c_i}{q_i^2} \right)    }$
 \\[12pt]
   $\displaystyle{\quad X =\frac 1{2(1+\lambda h_-)} \, X_+ +\frac {\omega^2(t)-\lambda h_+}{2(1+\lambda h_-)^2 }\,  X_-   }$
     \\[12pt]
 \hline

\hline
\end{tabular}
 \end{center}
\label{table2}
}
\end{table}

  \FloatBarrier
%%%%%%%%%%%%%%%%%%%%%%%%%%%%%%%%%%%

%%%%%%%%%%%%%%%%%%%%%%%%%%%%%%%%%%%%%%%%%%%%%%

 \subsubsection{Generalized  spherical and hyperbolic isotropic oscillators} 
\label{s912}

The most natural extension of the previous Euclidean   oscillators to proper nonlinear 
LH oscillator systems is to construct their counterpart on spaces of constant curvature:   the $n$-dimensional sphere  $\bS^n$ and hyperbolic space $\bH^n$.   To this aim, we consider   that $\>q$ are no longer Cartesian coordinates, but the Poincar\'e coordinates obtained through stereographic projection with ``south" pole, such that the metric (\ref{zb}) and scalar curvature (\ref{zc})  are given by~\cite{BH07,Doub}
\be
\dd s^2=4\, \frac{\dd \>q^2}{(1+\kappa \>q^2)^2},\qquad R=n(n-1)\kappa,
\label{metricP}
\ee
where the real parameter $\kappa$ is the constant sectional curvature of the space covering the sphere, hyperbolic and Euclidean spaces for positive, negative and zero values, respectively. 
For the hyperbolic space with $\kappa=-|\kappa|$, the projection  requires that $1-|\kappa| \>q^2>0$  
  leading to the interior of $n$-ball $|\kappa|\>q^2<1$; for $n=2$, it is   the Poincar\'e disk or the conformal disk model.

Observe that the (flat) Euclidean limit $\kappa=0$ of the metric (\ref{metricP}) gives 
$\dd s^2=4 \dd \>q^2$, thus the Poincar\'e coordinates are twice the usual Cartesian coordinates.
For this reason, we set $f=(1+\kappa h_-)^{-1}$ for the free  Hamiltonian (\ref{Th}), which is proportional to the conformal factor of the metric (\ref{metricP}), since it   yields $f\equiv 1$ when $\kappa=0$, whereas the $t$-dependent isotropic oscillator on $\bS^n$ and  $\bH^n$ in these coordinates reads  as~\cite{BH07}
 \be
\mU(t,|\>q|)=\frac 12 \omega^2({t})\frac{ h_- }{(1-\kappa h_-)^2}=\frac 12 \omega^2({t})\frac{\>q^2 }{(1-\kappa \>q^2)^2},
\label{HiggsP}
 \ee
reducing to (\ref{Usw}) on $\bE^n$ for $\kappa=0$. This potential can be better understood in terms of a radial distance $\rho$ measured along the geodesic joining the particle and the origin $O$ in the curved space, which does not coincide with the radial coordinate $|\>q|=r$; explicitly (see~\cite{annals2009}),  $\mU\propto \tan^2\!\rho $
on $\bS^n$ $(\kappa=+1)$ while $\mU\propto \tanh^2\!\rho $
on $\bH^n$ $(\kappa=-1) $. For the spherical case, it is worth observing that  (\ref{HiggsP}) is the $t$-dependent counterpart  of the so-called Higgs oscillator~\cite{Higgs},   widely studied in the literature (see~\cite{Pogoa, Nersessian1, ranran2, BaHeSantS03,Santander6, 2013Higgs} and references therein).

Then, when the $n$ arbitrary $c_i$-potentials are added to the main potential (\ref{HiggsP}), we obtain the generalized  spherical and hyperbolic isotropic oscillators with a $t$-dependent frequency shown in Table~\ref{table2}, which are also known as the ``curved SW systems"~\cite{Pogosyan1,RS,Kalnins2,BaHeSantS03,BH07,2013Higgs}. The  Demkov--Fradkin tensor for the corresponding $t$-independent Hamiltonians, which ensures maximal superintegrability, can be found explicitly in Poincar\'e variables in~\cite{BH07}. In addition, the $c_i$-potentials also behave as centrifugal barriers  on the hyperbolic space $\bH^n_0$ ($\>q\ne 0$), as in the previous Euclidean SW case on $\bE^n_0$, but they admit an interpretation as noncentral oscillators on the $n$-dimensional sphere $\bS^n_0$. In particular, let us set $\kappa=+1$ 
and  consider $n$ points $O_i$ on  $\bS^n_0$ such that $\{O,O_1,\dots,O_n\}$ are mutually separated by a quadrant, i.e., a distance $\frac \pi 2$. If $\rho_i$ denotes the distance along the geodesic joining the particle and the point $O_i$, then  the generalized spherical  isotropic oscillator can  be expressed as the superposition of the $t$-dependent  central Higgs oscillator with   frequency $  \omega(t)$ and centre at the origin $O$ and $n$ noncentral oscillators with    $t$-independent frequencies $c_i=  \omega^2_i/2$ and centre at $O_i$, namely~\cite{ranran2,BaHeSantS03,2013Higgs}
\be
\mU(t,|\>q|)= \frac 12  \omega^2({t}) \tan^2\!\rho +\frac 12 \sum_{i=1}^n \omega_i^2 \tan^2\!\rho_i.
\label{cpots}
\ee

All the results above can also be written in terms of  Beltrami coordinates corresponding to a  stereographic central projection~\cite{BH07}.   The resulting free Hamiltonian is no longer of conformal type, with metric (\ref{zb}),  but it is still   an $\mathfrak{sl}(2,\mathbb{R})$-coalgebra space (\ref{za}) for the functions $\mA$ and $\mB$ indicated in Table~\ref{table2}, whose underlying metric reads
\be
\dd s^2= \frac{(1+\kappa\>q^2)\dd\>q^2-\kappa
(\>q\scal \dd\>q)^2}{(1+\kappa\>q^2)^2} .
\nonumber
\ee
Recall that  there is again a restriction for the Beltrami coordinates on $\bH^n$: $|\kappa|\>q^2<1$;  for $n=2$, the hyperbolic plane is   the interior of the Klein disk. Note also that  in Beltrami variables,   the main potential (\ref{HiggsP}) has the same formal expression as in the Euclidean case (\ref{Usw}) for any value of the curvature $\kappa$ (the latter only appears within the kinetic term and $c_i$-potentials).

%%%%%%%%%%%%%%%%%%%%%%%%%%%%%%%%%%%%%%%%%%%%%%

 \subsubsection{Generalized   Darboux III oscillator} 
\label{s913}

Maximally superintegrable  systems in   spherically symmetric spaces (\ref{zb}) have been classified in~\cite{Bertrand2008} for $n=3$ from the results  obtained formerly in~\cite{Perlick}, which are called Bertrand Hamiltonians. Therefore, the previous spherical/hyperbolic oscillators are particular systems of this family, although, in general, the underlying spaces are of non-constant curvature (\ref{zc}) and the additional constants of the motion to the set (\ref{w10}) are usually of higher-order in the momenta. 
 
In addition to oscillators in constant curvature spaces,  the closest  curved oscillator to the Euclidean one is provided  by the so-called  Darboux III system~\cite{BEHR08a,BurgosAnnPh11}, whose underlying $n$-dimensional Darboux III space is defined by the following metric (\ref{zb}) and scalar curvature   (\ref{zc})
 \be
\dd s^2=  (1+\lambda\>q^2 )\dd\>q^2,\qquad R(|\>q|)= - \lambda (n-1)\,\frac{2n + 3\lambda(n-2)\>q^2}{(1+\lambda \>q^2)^3} ,
\nonumber
 \ee
 where $\lambda$ is an arbitrary real parameter. For negative values of $\lambda=-|\lambda|$, the domain of the variables is restricted to the ball $ |\lambda|\>q^2 <1$, reminding the hyperbolic space.
The $t$-dependent main potential in (\ref{Usw}) is given by
\be
\mU(t,|\>q|)=\frac 12 \omega^2({t}) \frac{h_- }{1+\lambda h_-}=\frac 12 \omega^2({t})
\frac{\>q^2}{1+\lambda \>q^2} ,
\nonumber
\ee
which reduces to the isotropic oscillator (\ref{Usw}) for $\lambda=0$. The generalized Darboux III  Hamiltonian  is then constructed by adding the $n$ $c_i$-potentials, for which the resulting expressions are displayed in Table~\ref{table2}. The  Demkov--Fradkin tensor  for the $t$-independent 
generalized system can be found in~\cite{BEHR08a}, which is still quadratic in the momenta.
Observe also that the whole system can naturally be   regarded as PDM oscillator system (see~\cite{CrNN07,Quesne07,CrR09,BurgosAnnPh11,GhoshRoy15,Quesne15Jmp}  and references therein), here  with $t$-dependent   frequency, with $m(|\>q|)=1+\lambda \>q^2$ and Hamiltonian (\ref{PDMg})  now reading as 
\be
h(t,\>q,\>p)=   \frac{1}{2m(|\>q|)}\,  \>p^2 +  \frac 12m(|\>q|)\omega^2(t)\frac{\>q^2}{(1+\lambda \>q^2)^2} + \frac{1}{2m(|\>q|)}\sum_{i=1}^n \frac{c_i}{q_i^2} ,
\nonumber
\ee
which entails a transformation of the oscillator potential according to the PDM.

%%%%%%%%%%%%%%%%%%%%%%%%%%%%%%%%%%%%%%%%%%%%%%

\subsection{Generalized curved KC  systems with  {\em t}-dependent coupling constant} 
\label{s92}

Along the same lines  as for the curved oscillators, we now construct three classes of $t$-dependent generalized KC systems in $n$ dimensions from maximally superintegrable $t$-independent Hamiltonians, but now with a  $t$-dependent coupling constant $\kk(t)$. 
   The main results comprising  their  Hamiltonian (\ref{gH}),   equations of the motion (\ref{eqs}) together with their associated vector field (\ref{gX}) are presented in Table~\ref{table3}.   Let us  briefly comment on these results. 
   
   \newpage

%%%%%%%%%%%%%%%% TABLE 3%%%%%%%%%%%%%%%%%%
  \FloatBarrier
   
 \begin{table}[ht!]
{\small
\caption{ \small{$t$-Dependent  generalized curved KC systems  from the Poisson $\mathfrak{sl}(2,\mathbb{R})$-coalgebra in Section~\ref{s8}, 
 with Hamiltonian functions $h_\alpha$ (\ref{w6}) and vector fields $X_\alpha$ (\ref{w8})  $(\alpha=-,+,3)$.  For each case, the $t$-dependent Hamiltonian, equations of the motion and vector field are given; for all of them, $\>\prod_{i=1}^n|q_i|^{|c_i|}\ne 0$ and  $i=1,\dots,n$.}}
  \begin{center}
\noindent 
\begin{tabular}{l}
\hline

\hline
\\[-6pt]
$\bullet$ Generalized Euclidean KC system:  \quad  $\ff =1 $\qquad $\mU = -\kk(t) h_-^{-1/2}$ 
\\[4pt]
   $\displaystyle{\quad h =\frac 12h_+ -\frac {\kk(t)} {\sqrt{h_-} }=  \frac{1}{2}\>p^2 -\frac {\kk(t)} { |\>q|}   +\frac 12\sum_{i=1}^n \frac{c_i}{q_i^2} }$
   \\[12pt]
 $\displaystyle{\quad   \frac{\dd q_i}{\dd t}= p_i \qquad   \frac{\dd p_i}{\dd t}=- \kk(t)\frac{q_i}{|\>q|^3}  +\frac{c_i}{q_i^3} }$\\[10pt]
   $\displaystyle{\quad X =\frac 12X_+ +\frac {\kk(t)}{2h_-^{3/2} }\, X_-   }$
  \\[12pt]
  \hline
\\[-6pt]
 
$\bullet$ Generalized spherical and hyperbolic KC systems \\[4pt]
  {\em (i) Poincar\'e variables} (curvature $\kappa$):   \quad  $\ff =(1+\kappa h_-)^{-1} $\qquad $\mU = -\kk(t)(1-\kappa h_-)   h_-^{-1/2}$
\\[4pt]
    $\displaystyle{\quad h =\frac 12(1+\kappa h_-)^{2}h_+ -\kk(t)\frac{1-\kappa h_- } {\sqrt{h_-} } }$
   \\[10pt]
      $\displaystyle{\qquad\!  =  \frac{1}{2}(1+\kappa \>q^2)^{2}\>p^2 -\kk(t)\frac{1-\kappa \>q^2 } {|\>q|}  + \frac{1}{2}(1+\kappa \>q^2)^{2}\sum_{i=1}^n \frac{c_i}{q_i^2} }$
   \\[12pt]
 $\displaystyle{\quad   \frac{\dd q_i}{\dd t}=(1+\kappa \>q^2)^{2} p_i \qquad   \frac{\dd p_i}{\dd t}=- \kk(t)\frac{q_i (1+\kappa \>q^2)}{ |\>q|^3} +\frac{c_i(1+\kappa \>q^2)^{2}}{q_i^3} - 2\kappa q_i (1+\kappa \>q^2) \!\left( \!\>p^2+\sum_{i=1}^n  \frac{c_i}{q_i^2} \right)  }$
 \\[12pt]
   $\displaystyle{\quad X =\frac 12(1+\kappa h_-)^{2}X_+ + \left(\! \kk(t)\frac {1+\kappa h_-}{2 h_-^{3/2}}  +\kappa (1+\kappa h_-)h_+ \!\right)\!X_-  }$
    \\[14pt]
   
     {\em  (ii) Beltrami variables} (curvature $\kappa$):   \quad  $\mA =\frac12(1+\kappa h_-)  $\qquad$\mB =\frac12 \kappa(1+\kappa h_-)  $\qquad $\mU =-\kk(t) h_-^{-1/2} $
   \\[4pt]
    $\displaystyle{\quad h =\frac 12(1+\kappa h_-) \bigl( h_+ + \kappa h_3^2 \bigr) -\frac {\kk(t)} {\sqrt{h_-} } }$
   \\[8pt]
      $\displaystyle{\qquad\!  =  \frac{1}{2} (1+\kappa \>q^2 ) \bigl(\>p^2 + \kappa  ( \>q\scal  \>p )^2 \bigr)-\frac {\kk(t)} { |\>q|} + \frac{1}{2}(1+\kappa \>q^2) \sum_{i=1}^n \frac{c_i}{q_i^2} }$
   \\[8pt]
 $\displaystyle{\quad   \frac{\dd q_i}{\dd t}=(1+\kappa \>q^2) \bigl( p_i  + \kappa q_i ( \>q\scal  \>p )\bigr)   }$
 \\[8pt]
 $\displaystyle{\quad      \frac{\dd p_i}{\dd t}=- \kk(t)\frac{q_i}{|\>q|^3}  +(1+\kappa \>q^2)\left(\! \frac{c_i }{q_i^3}-\kappa p_i  ( \>q\scal  \>p )\! \right)-  \kappa q_i  \!\left(   \!\>p^2+ \kappa ( \>q\scal  \>p )^2 + \sum_{i=1}^n  \frac{c_i}{q_i^2} \right)  }$
 \\[12pt]
   $\displaystyle{\quad X =\frac 12(1+\kappa h_-) X_+ +  \kappa (1+\kappa h_-) h_3 X_3 +\frac 12\! \left(  {\kk(t) }{h_-^{-3/2}} +\kappa \bigr(h_+ +\kappa h_3^2 \bigl) \right)\! X_-  }$
       \\[10pt]
  \hline
\\[-6pt]

$\bullet$ Generalized Taub--NUT KC system:   \quad  $\ff = \sqrt{1+ {\eta}{  h_-^{-1/2}} }$\qquad $\mU = -\kk(t) \bigl(\eta+  \sqrt{h_- }\,\bigr)^{-1} $ 
\\[6pt]
   $\displaystyle{\quad h =\frac  { \sqrt{h_- } \, h_+}{2\bigl(\eta+  \sqrt{h_- }\,\bigr)} - \frac  {\kk(t)}{\eta+  \sqrt{h_- }}  =  \frac{|\>q|}{\eta+  |\>q|}\left(\frac 12 \>p^2 -\frac{\kk(t)}{|\>q|}  +\frac 12\sum_{i=1}^n \frac{c_i}{q_i^2} \right)}$
   \\[12pt]
 $\displaystyle{\quad   \frac{\dd q_i}{\dd t}= \frac{|\>q| p_i}{\eta+  |\>q|} \qquad   \frac{\dd p_i}{\dd t}=- \kk(t)\frac{q_i}{ |\>q| ( \eta+  |\>q| )^2} +\frac{c_i|\>q|}{q_i^3( \eta+  |\>q| )} - \frac{\eta q_i}{2  |\>q| ( \eta+  |\>q| )^2}    \!\left( \!\>p^2+\sum_{i=1}^n  \frac{c_i}{q_i^2} \right)    }$
 \\[12pt]
   $\displaystyle{\quad X =\frac  { \sqrt{h_- }  }{2\bigl(\eta+  \sqrt{h_- }\,\bigr)} \, X_+ +\frac {2 \kk(t)+\eta h_+}{4\sqrt{h_- } \bigr(\eta+  \sqrt{h_- } \,\bigl)^2 }\,  X_-   }$
     \\[16pt]
 \hline

\hline
\end{tabular}
 \end{center}
\label{table3}
}
\end{table}

  \FloatBarrier
%%%%%%%%%%%%%%%%%%%%%%%%%%%%%%%%%%%

  The first system is the usual KC system on the Euclidean space with a $t$-dependent main potential   (\ref{gH}) given by
\be
\mU(t,|\>q|)=-\frac{\kk(t) }{\sqrt{h_- }}=-\frac{\kk(t) }{|\>q|} ,
\nonumber
\ee
    which, in contrast to the oscillator (\ref{Usw}), is not a proper LH system but a genuine nonlinear LH system. The generalized KC system is obtained by adding the $n$ $c_i$-potentials, which behave as centrifugal barriers as described for the SW system in Section~\ref{s911}.      When at least  a single  $c_i$-term vanishes, the $t$-independent Hamiltonian  admits, besides (\ref{w10}), an additional ``hidden" constant of the motion determined by a   Laplace--Runge--Lenz component, which is quadratic in the momenta~\cite{Evans90b,Miguel,Williamsx}. Nevertheless, when all $c_i\ne0$, all Laplace--Runge--Lenz constants of the motion are  quartic in the momenta~\cite{Verrier2008,KC2009}.

 The $t$-dependent flat generalized KC system can be extended to the sphere $\bS^n$ and hyperbolic space $\bH^n$, with constant sectional curvature $\kappa$, by using  the Poincar\'e or Beltrami projective variables considered in Section~\ref{s912} for the spherical/hyperbolic oscillators, as shown  in Table~\ref{table3}.  In these coordinates, the KC potential  turns out to be~\cite{BH07}
 \be
 \begin{aligned}
 \mbox{Poincar\'e:}&\quad  \mU(t,|\>q|)=-\kk({t})\frac{1-\kappa h_-}{ \sqrt{h_-} }=-\kk({t})\frac{1-\kappa \>q^2}{ |\>q| } .\\[2pt]
 \mbox{Beltrami:}&\quad \mU(t,|\>q|)=-\frac{\kk({t})}{ \sqrt{h_-} }=-\frac{\kk({t})}{ |\>q| } . 
  \end{aligned}
  \nonumber
 \ee    
Similarly to the   isotropic (Higgs) oscillator (\ref{HiggsP}), by  introducing a geodesic radial distance $\rho$, it is found that the spherical/hyperbolic KC potential becomes $\mU\propto 1/\tan \rho $
on $\bS^n$ $(\kappa=+1)$ and  $\mU\propto 1/\tanh \rho $
on $\bH^n$ $(\kappa=-1) $ ~\cite{annals2009,car2005}.   As in the (flat) Euclidean KC system, when
  the  $c_i$-potentials are added, it is well known that all Laplace--Runge--Lenz constants of the motion are quartic in the momenta~\cite{KC2009}, but whenever  one or more $c_i$ vanish, there exists a  component that is quadratic  in the momenta ~\cite{BH07,RS,Kalnins2}.    
  The interpretation of the $c_i$-potentials is exactly the same as in the curved SW system in Section~\ref{s912}. In the hyperbolic space $\bH^n_0$ all of them yield centrifugal barriers, but on the sphere
   $\bS^n_0$ they can be seen as noncentral oscillators  whenever $c_i=\omega^2_i/2>0$, as in (\ref{cpots}), thus determining a combination of the $t$-dependent curved KC potential with such $n$-oscillators in the form $(\kappa=+1)$:
 \be
\mU(t,|\>q|)=-\frac{ \kk({t})}{ \tan \rho  }   +\frac 12 \sum_{i=1}^n \omega_i^2 \tan^2\!\rho_i.
\nonumber
\ee

The last application concerns the so-called Taub--NUT KC system~\cite{NUT,LatiniNUT}, which is 
another $t$-dependent Bertrand Hamiltonian~\cite{Bertrand2008}.  The metric (\ref{zb}) and scalar curvature   (\ref{zc}) of the corresponding  Taub--NUT  space  are given by  
\be
\dd s^2=  \left(1+\frac{\eta}{|\>q|} \right)\dd\>q^2,\qquad R(|\>q|)=   \eta (n-1)\,\frac{4(n-3)|\>q|+ 3 \eta(n-2)  }{4|\>q|(\eta+ |\>q|)^3} ,\qquad  |\>q|\ne 0,
\nonumber
 \ee
 where $\eta$ is an arbitrary real parameter. Then, if $\eta=-|\eta|$, the domain of the coordinates  is limited to the exterior of the ball $|\>q| > |\eta|$.
The $t$-dependent curved KC potential in (\ref{Usw}) reads
\be
\mU(t,|\>q|)=-\frac{
\kk ({t})  }{\eta+ \sqrt{h_-}}=-\frac{
\kk ({t})  }{\eta+  |\>q|},
\nonumber
\ee
whose $t$-independent counterpart is endowed with a Laplace--Runge--Lenz $n$-vector quadratic in the momenta.  The corresponding generalized system is displayed in Table~\ref{table3} and, to the best of our knowledge, the Laplace--Runge--Lenz constants of the motion have not yet been derived for the $t$-independent generalized Hamiltonian.  Similarly, other $t$-dependent Bertrand Hamiltonians~\cite{Bertrand2008} can be constructed.

%%%%%%%%%%%%%%%%%%%%%%%%%%%%%%%%%%%%%%%%%%%

\section{Conclusions and outlook}
\label{Sec:Conc}

In this work, we have provided a nonlinear generalization of the notion of LH systems, that can also be understood as a $t$-dependent generalization of collective Hamiltonians, as well as a $t$-dependent extension of systems admitting a Poisson coalgebra symmetry. The main properties of nonlinear LH systems have been studied, pointing out that these constitute far reaching generalizations of properties relative to LH systems. In addition, the standard Poisson coalgebra symmetry method has also been extended to cover the case of $t$-dependent Hamiltonian systems. In this context, the relevance of nonlinear LH systems emerges simultaneously from the theory and the wide range of applications, that allow us to generalize several well-known systems to include explicit   $t$-dependence, still preserving some of their salient structural properties. 

As some relevant applications we have obtained, besides $t$-dependent HH Hamiltonians and Painlev\'e trascendents, the most general $t$-dependent central potential on an $n$-dimensional
spherically symmetric space. From these, several $t$-dependent curved oscillators and KC Hamiltonians whose $t$-independent counterparts have been extensively considered in the literature are deduced. In this respect, it is clear that each of these systems should be studied separately,  analyzing   the properties  that  a specific  frequency   $\omega(t)$ or coupling constant $\kk(t)$ entails, as well as studying the corresponding equations of the motion. These systems, with underlying Riemannian spaces and of considerable physical interest,  are just a glimpse of the wide potentialities and applicability of this novel $t$-dependent nonlinear formalism. 

Among the large number and variety of situations where this technique may be applied, we merely mention a few of them for their physical relevance:
\begin{itemize}
  
\item As all Painlev\'e equations are known to be Hamiltonian, the techniques proposed in this work 
can be applied to them. In this context, it is worthy to be inspected whether new information on the fundamental characteristics and solutions of the Painlev\'e equations can be obtained from the nonlinear approach.

\item A $t$-dependence on the curvature for a given  $\mathfrak{sl}(2,\mathbb{R})$-coalgebra space  can also be implemented,  leading to a $t$-dependent geodesic motion with an underlying $t$-dependent metric. For instance, this can be done through $t$-dependent parameters $\kappa(t)$, $\lambda(t)$ and $\eta(t)$ in the systems studied here.  This does not exhaust the possibilities, as  $t$-dependent potentials can also be added.

\item Another research line worthy to be explored is the application of the theory to the framework of gravity and cosmology. This requires, first,   constructing Lorentzian spaces and, second, potentials from a Poisson coalgebra different from $\mathfrak{sl}(2,\mathbb{R})$. Natural candidates in this context are the Lorentz algebra $\mathfrak{so}(3,1)$ and the symplectic Lie algebra $\mathfrak{sp}(2n,\mathbb{R})$,  as well as some of their distinguished Lie subalgebras. In fact, LH systems from $\mathfrak{so}(3,1)$ and $\mathfrak{sp}(4,\mathbb{R})$  have been recently constructed in~\cite{CCH25}, leading to coupled $t$-dependent oscillators on the Minkowskian plane.  
 
\end{itemize}

To conclude, we observe that nonlinear LH systems still admit a further extension. Explicitly, the Poisson coalgebra deformation theory of LH systems \cite{BCFHL21} is concerned with Hamiltonian functions that depend on a finite family of functions $ h_{1,z},\ldots,h_{r,z} $ such that 
$$
\{h_{\alpha,z},h_{\beta,z}
\}=\mathcal{F}_{z,\alpha\beta}(h_{1,z},\ldots,h_{r,z}),\qquad 1\leq \alpha<\beta\leq r,\qquad z\in \mathbb{R},
$$
where $z$ is a quantum deformation parameter related to the usual one as $q=\exp z$.  It turns out that some parts of this formalism can be adapted to the case of nonlinear LH systems, leading to a notion of nonlinear deformed LH systems, by taking into account some minor technical changes and adjustments. The most relevant technical question makes reference to the momentum map $J$, which must be replaced by a more general compatible geometric object, whose precise nature must still be analyzed in detail.  Work in these various directions is currently in progress, and we expect to address some progress on these topics in some future work.

%%%%%%%%%%%%%%%%%%%%%%%%%%%%%%%%%%%%%%%%%%%

\subsection*{Acknowledgements}

\phantomsection
\addcontentsline{toc}{section}{Acknowledgements}

{\small 
R.C.-S.~and F.J.H.~have been partially supported by Agencia Estatal de Investigaci\'on (Spain) under  the grant PID2023-148373NB-I00 funded by MCIN/AEI/10.13039/501100011033/FEDER, UE.  F.J.H.~acknowledges support  by the  Q-CAYLE Project  funded by the Regional Government of Castilla y Le\'on (Junta de Castilla y Le\'on, Spain) and by the Spanish Ministry of Science and Innovation (MCIN) through the European Union funds NextGenerationEU (PRTR C17.I1).   J. de L.~acknowledges a CRM-Simons professorship funded by the Simons Foundation and the Centre de Recherches Math\'ematiques (CRM) of the Universit\'e Montr\'eal. The authors also acknowledge the contribution of RED2022-134301-T funded by MCIN/AEI/10.13039/ 501100011033 (Spain).
}

%%%%%%%%%%%%%%%%%%%%%%%%%%%%%%%%%%%%%%%%%%%

\end{document}